\documentclass[journal, 10pt]{IEEEtran}
\usepackage{graphics}
\usepackage{graphicx}
\usepackage{amsfonts}
\usepackage{multirow}
\usepackage{epstopdf}
\usepackage{color}
\usepackage{float}
\usepackage[font=small]{caption}
\usepackage{bm}
\usepackage{algorithm}
\usepackage{algorithmic}
\usepackage{amsmath}
\usepackage{pifont}
\usepackage[utf8]{inputenc}
\usepackage{mathtools}
\usepackage{amssymb}
\usepackage{amsthm}
\usepackage{lipsum}% http://ctan.org/pkg/lipsum
\usepackage{bbm}
\usepackage{enumitem}   
\usepackage{cite}
\usepackage{steinmetz} % For Phase Symbol
\usepackage{placeins} % For FloatBarrier
\usepackage{dsfont} %\mathds
\usepackage{cuted}
\usepackage{wasysym} % Checked Mark - Nike
\usepackage[table]{xcolor}
\newcolumntype{s}{>{\columncolor[HTML]{FE6F5E}} c}
\newcolumntype{t}{>{\columncolor[HTML]{5D8AA8}} c}
\definecolor{LightRed}{rgb}{1,0.85,1}
\definecolor{LightBlue}{rgb}{0.8,0.9,1}
\definecolor{LightGreen}{rgb}{0.8,1,0.8}

\usepackage[colorinlistoftodos]{todonotes} % Todo

\usepackage{booktabs}% http://ctan.org/pkg/booktabs

\newtheorem{theorem}{Theorem}

\newtheorem{lemma}{Lemma}
\theoremstyle{definition}
\newtheorem{definition}{Definition}
\newtheorem{criterion}{Criterion}

\newcommand{\diag}{\text{diag}}
\newcommand{\bb}[1]{\mathbb{#1}}
\newcommand{\mds}[1]{\mathds{#1}}
\newcommand{\abs}[1]{\left|#1\right|}
\newcommand{\norm}[1]{\left|\left|#1\right|\right|}
\newcommand{\defeq}{\triangleq}
\newcommand{\Conv}[1]{%
  \mathop{\scalebox{1.5}{\raisebox{-0.6ex}{\makebox[0.5\width][r]{$\circledast $}}}}\limits^{#1}_{i=0}
}

\ifCLASSINFOpdf
\else
\fi

\begin{document}

\title{Sparse Array Design via Fractal Geometries}

\author{Regev~Cohen,~\IEEEmembership{Graduate Student,~IEEE,}
        Yonina~C.~Eldar,~\IEEEmembership{Fellow,~IEEE}
\thanks{R. Cohen (e-mail:  regev.cohen@gmail.com) is with the Department of Electrical Engineering, Technion-Israel Institute of Technology, Haifa 32000, Israel. Y. C. Eldar (e-mail: yonina.eldar@weizmann.ac.il) is with the Faculty of Math and Computer Science,  Weizmann Institute of Science, Rehovot, Israel. This project has received funding from the
European Union’s Horizon 2020 research and innovation program under grant
No. 646804-ERC-COG-BNYQ, and from the Israel Science Foundation under
grant No. 0100101.}}

\maketitle

\begin{abstract}
Sparse sensor arrays have attracted considerable attention in various fields such as radar, array processing, ultrasound imaging and communications. 
In the context of correlation-based processing, such arrays enable to resolve more uncorrelated sources than physical sensors. This property of sparse arrays stems from the size of their difference coarrays, defined as the differences of element locations. Thus, the design of sparse arrays with large difference coarrays is of great interest. In addition, other array properties such as symmetry, robustness and array economy are important in different applications. Numerous studies have proposed diverse sparse geometries, focusing on certain properties while lacking others. Incorporating multiple properties into the design task leads to combinatorial problems which are generally NP-hard. For small arrays these optimization problems can be solved by brute force, however, in large scale they become intractable. In this paper, we propose a scalable systematic way to design large sparse arrays considering multiple properties. To that end, we introduce a fractal array design in which a generator array is recursively expanded according to its difference coarray. Our main result states that for an appropriate choice of the generator such fractal arrays exhibit large difference coarrays. Furthermore, we show that the fractal arrays inherit their properties from their generators. Thus, a small generator can be optimized according to desired requirements and then expanded to create a fractal array which meets the same criteria. This approach paves the way to efficient design of large arrays of hundreds or thousands of elements with specific properties.         
\end{abstract}


\section{Introduction}
\label{sec:intro}
\IEEEPARstart{S}{ensor} array design plays a key role in various fields such as radar \cite{pal2010nested},
% change 4 \cite{merrill2001introduction,vaidyanathan2011sparse,pal2015pushing,pal2012correlation,amin2015special,vaidyanathan2011theory,pal2012nested},
radio \cite{bracewell1962radio}, communications \cite{van2002optimum} and ultrasound imaging \cite{cohen2018sparse,cohen2017sparse,cohen2018optimized}. In particular, two major applications of array processing \cite{tan2014direction,eldar2007competitive,eldar2006expected} are direction-of-arrival (DOA) estimation and beamforming used for detecting sources impinging on an array. Such applications often utilize sparse arrays, namely sensor arrays with non-uniform element spacing, since under certain conditions they allow to identify more uncorrelated sources than physical sensors \cite{qiao2019guaranteed, koochakzadeh2018fundamental, qiao2017maximum}. 

Sparse arrays include the well-known minimum redundancy arrays (MRA) \cite{moffet1968minimum}, minimum holes arrays (MHA), nested arrays (NA) \cite{pal2010nested} and coprime arrays (CP) \cite{pal2011coprime}, to name just a few. Such arrays enable detection of $\mathcal{O}(N^2)$ sources using $N$ elements, unlike uniform linear arrays (ULAs) that can resolve $\mathcal{O}(N)$ targets. This property of increased degrees-of-freedom (DOF) relies on correlation-based processing and arises from the size of the difference coarray, defined as the pair-wise sensor separation. A large difference coarray increases resolution \cite{moffet1968minimum,pal2010nested,pal2011coprime} and the number of resolvable sources
\cite{pal2010nested,pal2011coprime}.
%  change 5 \cite{pal2010nested,pal2011coprime,liu2017cramer}.
Hence, the size of the difference set is an important metric in designing sparse arrays.       

Other array properties are also important in specific applications. For example, typically it is required that the sensor locations be expressed in closed-form to enable simple and scalable array constructions. Array symmetry is also often favorable as it reduces complexity \cite{haupt1994thinned,friedlander1991direction} and improves performance \cite{xu1992detection,xu2006deflation,ye2007doa}. A contiguous difference coarray facilitates the use of standard DOA recovery algorithms \cite{pal2010nested}. The array weight function and beampattern govern the array performance, therefore, it is convenient if they can be expressed in simple forms that allow analysis and optimization. Since electromagnetic element interaction may lead to adverse effects on the beampattern, an array with low mutual coupling \cite{liu2016super, boudaher2015doa}
% change 1 \cite{liu2016superconf,liu2016super,liu2016high,liu2017hour,friedlander1991direction,hui2007decoupling,pasala1994mutual,hui2003improved,svantesson1999modeling,svantesson2000mutual,lin2006blind,sellone2007novel,ye2009doa,lui2010mutual,boudaher2015doa}
is beneficial. To reduce power and cost, the array should be economic where all sensors are essential \cite{liu2017maximally}. Conversely, elements may malfunction, in which case redundancy increases the array robustness to sensor failures \cite{liu2018robust,liu2018compare,liu2018optimizing}.  

The introduction of nested arrays \cite{pal2010nested} and coprime arrays \cite{pal2011coprime} has sparked great interest in non-uniform arrays, leading to  
numerous studies proposing diverse sparse configurations. The authors of \cite{liu2016super, liu2016superII} introduced variants of nested arrays, called super nested arrays (SNA), which redistribute the elements of the dense ULA part of the nested array to obtain reduced mutual coupling.
In \cite{liu2017augmented}, augmented nested arrays (ANAs) are created by splitting the dense ULA of a nested array into two or four parts that can be relocated to two sides of the sparse ULA of a nested array. This leads to increased DOF and reduced mutual coupling compared with nested and super nested arrays. To allow robustness to sensor failures, robust MRAs (RMRAs) are presented in \cite{liu2018optimizing} where the sensor locations are given by a combinatorial problem in which the array fragility \cite{liu2018robust} is constrained. Alternative designs include utilizing array motion \cite{qin2019doa} to fill in the holes in the difference coarray of a given sparse array (e.g. coprime), and array configurations based on the maximum inter-element spacing constraint (MISC) \cite{zheng2019misc} that achieve more DOF than nested arrays and ANAs while exhibiting less mutual coupling than super nested arrays.   
% Many more sparse geometries are available, however, it is beyond the scope of this paper to review this rich literature.

In addition to the above, several variants of coprime arrays have been introduced. A generalization of coprime arrays named coprime array with displaced subarrays
(CADiS) is developed in \cite{qin2015generalized}. CADiS is built by compressing the inter-element
spacing of one subarray of the coprime array while displacing the other subarray, yielding an array with higher number of unique lags in the difference coarray and reduced mutual coupling. A thinned coprime array (TCA) is presented in \cite{raza2017thinned} where some of the sensors of an extended coprime array \cite{pal2011coprime} are removed without affecting the aperture and the difference coarray, resulting in lower mutual coupling. Complementary coprime arrays (CCP) are derived in \cite{wang2019hole} where a dense ULA is added to an extended coprime array to ensure contiguous difference coarray. In \cite{zheng2019extended}, part of the sensors in extended coprime arrays and CADiS are repositioned to create sliding extended coprime arrays (SECAs) and relocating extended coprime arrays (RECAs) receptively, which offer increased DOF and reduced mutual coupling compared with the original arrays.

Most existing array configurations focus on certain properties while lacking or being indifferent to others. For example, MRA and MHA yield large difference coarrays but are not expressed in closed-form. 
Nested and coprime arrays have simple forms but the former suffers from high mutual coupling while the latter exhibit holes in the difference coarray. Super nested arrays and ANAs enjoy reduced mutual coupling but are not robust to sensor failures. TCA offers increased DOF compared to coprime arrays but lacks a contiguous difference coarray. In contrast, complementary coprime arrays and MISC demonstrate contiguous difference coarrays but are susceptible to sensor failures. Moreover, none of the above are symmetric. 

There exists a broad range of applications which utilize sparse sensor arrays, each with different requirements. Hence, the array design must consider multiple properties and the appropriate tradeoffs between them, depending on the specific setting. Moreover, as technology advances, applications such as massive multiple-input multiple-output (MIMO) communications utilize an increasing number of sensors, expected to reach several thousands in the near future \cite{bjornson2019massive}. In such settings, the array construction has to be simple and efficient to allow scalability. The design process can often be formulated as an optimization problem which incorporates all the required array specifications. However, unfortunately, these design problems are combinatorial in nature, hence, they are intractable in large scale. 
Consequently, the development of a systematic and scalable approach to design large sparse arrays with multiple desired properties is of increasing importance.

The main contribution of this paper is in introducing a fractal design approach for sparse arrays which is scalable and considers multiple array properties. Fractal arrays
\cite{puente1996fractal,werner1999fractal,werner2003overview,feder2013fractals,falconer2004fractal} are geometrical structures which display an inherent self-similarity over different scales, and hence are used in the design of radiating systems to allow mutliwavelength operation \cite{puente1996fractal}. The construction of fractal structures is performed by recursively scaling a basic array, known as the generator. Here, we derive a special type of fractal arrays in which the generator is taken to be a sparse array and the scaling/translation factor is directly determined by the generator's difference coarray. We prove that an appropriate choice of the generator leads to fractal arrays with sizable difference coarrays. We then study the properties of the resultant fractal arrays, showing that they inherit their properties from the generators. In particular, a sparse fractal array exhibits the same coupling leakage as the generator, similar increased DOF, and it is at least as robust as the generator. 

Our proposed framework allows to extend any known sparse configuration to a large array while preserving its properties. It can be seen as a generalization of Cantor arrays that achieves increased DOF. Furthermore, a small sparse array can be designed and optimized according to given requirements, and then recursively expanded to generate an arbitrarily large array which meets the same design criteria. Thus, we establish a simple systematic approach for large sparse array design which incorporates multiple favorable properties. 

The paper is organized as follows. Section\,\ref{sec:preliminaries} reviews preliminaries of sparse arrays, including design criteria, and formulates the problem. We introduce our proposed fractal arrays in Section\,\ref{sec:fractal} and study their properties in Section\,\ref{sec:properties}. Section\,\ref{sec:experiments} provides numerical experiments of large array designs and performance analysis. Finally, Section\,\ref{sec:conclude} concludes the paper.    

\section{Review of Sparse Arrays}
\label{sec:preliminaries}

\subsection{Signal Model}
\label{subsec:model}

Consider $K$ narrowband sources with carrier wavelength $\lambda$ impinging on an $N$ element linear array.  The array sensors are located at $d_n=n\lambda/2$ where $n$ belongs to an integer set $\bb{G}$ ($|\bb{G}|=N$). For brevity, we refer to such an array as a linear array $\bb{G}$. Let $s_i\in\mds{C}$ and $\theta_i\in[-\pi/2,\pi/2]$ be the complex amplitude and DOA of the $i$th source respectively. Neglecting mutual coupling \cite{svantesson2000mutual}, the array measurement vector $\bf x$ can be expressed as
\begin{equation}
\label{eq:model}
{\bf x}=\sum_{i=1}^K s_i{\bf a}(\theta_i)+{\bf w}={\bf As}+{\bf w}\in\mds{C}^N,
\end{equation}
where ${\bf s}=[s_1\;s_2\cdots s_K]^T\in\mds{C}^K$, ${\bf A}=[{\bf a}(\theta_1)\;{\bf a}(\theta_2)\cdots{\bf a}(\theta_K)]\in\mds{C}^{N\times K}$ is the array manifold and ${\bf a}(\theta)\in\mds{C}^{N\times 1}$ is a steering vector at direction $\theta$ whose entries are $e^{j\pi\sin(\theta)n}$ ($n\in\bb{G}$). Here $\bf w$ represents additive white noise. We assume the source vector $\bf s$ and the noise $\bf w$ are zero-mean and uncorrelated, so that
\begin{itemize}
\item $\mds{E}[{\bf s}]=0,\;\mds{E}[{\bf w}]=0$,
\item $\mds{E}[{\bf ws}^H]=0$,
\item $\mds{E}[{\bf ss}^H]=\diag(p_1,p_2,...,p_K),\; \mds{E}[{\bf ww}^H]=p_w{\bf I}_N$,
\end{itemize}
where $p_i$ and $p_w$ are the power of the $i$th source and the noise respectively. We denote by ${\bf I}_N$ the $N\times N$ identity matrix.

The covariance of $\bf x$ can be written as
\begin{equation}
\label{eq:acmatrix}
{\bf R_x}=\sum_{i=1}^K p_i{\bf a}(\theta_i){\bf a}(\theta_i)^H+p_w{\bf I}_N.
\end{equation}
Vectorizing $\bf R_x$ to a vector $\bf r_x$ and averaging duplicate entries we obtain
\begin{equation}
\label{eq:acmodel}
{\bf r_x}=\sum_{i=1}^K p_i{\bf b}(\theta_i)+p_n{\boldsymbol\delta}={\bf Bp}+p_w{\boldsymbol\delta}\in\mds{C}^{|\bb{D}|},
\end{equation}
where $\bb{D}$ is the difference coarray defined below, ${\boldsymbol\delta}\in\mds{C}^{|\bb{D}|}$ is the Kronecker delta and the steering vectors are ${\bf b}(\theta)={\bf a}^\ast(\theta)\odot{\bf a}(\theta)$ where superscript $^\ast$ denotes conjugation and $\odot$ is the Khatri-Rao product \cite{van2002optimum,ma2009doa}.
The vector $\bf r_x$ can be seen as the signal received by a virtual array whose manifold is ${\bf B}=[{\bf b}(\theta_1)\;{\bf b}(\theta_2)\cdots{\bf b}(\theta_K)]\in\mds{C}^{\abs{\bb{D}}\times K}$ and its sensor locations are determined by the difference coarray $\bb{D}$.
\begin{definition}[Difference Coarray]
\label{def:diff}
Consider a sensor array $\bb{G}$. The difference set of $\bb{G}$ is given by 
\begin{equation*}
\bb{D}\defeq\{d\,|\;n_1-n_2=d,\,n_1,n_2\in\bb{G}\}.
\end{equation*}
The \textit{difference coarray} of an array $\bb{G}$ is the linear array $\bb{D}$.
The DOF of a linear array $\bb{G}$ is the cardinality of its difference coarray $\bb{D}$.
\end{definition}
\noindent The performance of correlation-based estimators is governed by the DOF of the sparse array. When the difference coarray is larger than the physical array, we can recover more uncorrelated targets than sensors or alternatively increase the angular resolution of DOA estimation \cite{moffet1968minimum,pal2012nested,pal2011coprime}.

\subsection{Difference Coarray Criteria}
Our goal is to generate arrays with increased DOF. To that end, we outline several popular criteria in sparse array design. We begin with a few  definitions related to sparse arrays.      

\begin{definition}[Central  ULA]
\label{def:center}
Consider a sensor array $\bb{G}$ with a difference coarray $\bb{D}$. Given a non-negative integer $m$, let $\bb{U}_m\defeq\{-m,...,-1,0,1,...,m\}$.
The \textit{central ULA} of $\bb{D}$ is the ULA defined as
\begin{equation*}
\bb{U}\defeq\; \underset{\bb{U}_m\subseteq\bb{D}}{\arg\max}\; |\bb{U}_m|.
\end{equation*}
The central ULA $\bb{U}$ is the maximum contiguous ULA that includes the 0th element in the difference coarray. 
\end{definition}

\begin{definition}[Hole-Free/Contiguous Difference Coarray]
Consider a sensor array $\bb{G}$ whose difference coarray is $\bb{D}$ and denote the central ULA of $\bb{D}$ by $\bb{U}$. The difference coarray $\bb{D}$ is said to be hole-free (i.e. contiguous) if $\bb{D}=\bb{U}$.
\end{definition}

\noindent Equipped with the definitions above, we state the following criteria for sparse array design \cite{liu2017maximally}:
\begin{criterion}[Closed-form]
For scalability, the sensor locations should be expressed in closed-form.
\label{criterion:closedform}
\end{criterion}

\begin{criterion}[Hole-free difference coarray]
For subspace-based DOA estimation
methods, e.g. MUSIC \cite{schmidt1986multiple} and ESPRIT \cite{roy1989esprit}, the number of
resolvable sources is determined by the cardinality of the central ULA. Hence, to exploit the difference coarray to its fullest extent, we require it to be hole-free \cite{liu2017maximally}.
\label{criterion:holefree}
\end{criterion}

\begin{criterion}[Large difference coarray]
To achieve increased DOF with respect to the number of sensors, the number of virtual elements (unique lags) of the difference coarray should satisfy $\abs{\bb{D}}=\mathcal{O}(\abs{\bb{G}}^2)$.    
\label{criterion:largearray}
\end{criterion}

Interestingly, most known array geometries do not fulfill Criteria\,\ref{criterion:closedform} to \ref{criterion:largearray} concurrently. MRAs \cite{moffet1968minimum}, RMRAs \cite{liu2018optimizing} and MHAs do not have a closed-form expression. ULAs and Cantor arrays \cite{liu2017maximally} exhibit difference coarrays whose size is $\mathcal{O}(N)$ and $\mathcal{O}(N^{\log_23})$, respectively. The difference coarray of a co-prime array is not hole-free, hence, interpolation may be required which increases complexity \cite{tuncer2007direction,liu2016coprime,qiao2017unified}.   
Note that nested arrays do meet the discussed requirements, however, they lack other important array properties. For example, they contain a dense ULA which results in high mutual coupling. To circumvent this limitation, super-nested arrays \cite{liu2016super} were introduced. However, the latter are expressed in a closed yet complicated form. Moreover, both nested arrays and super-nested arrays are not symmetric and are sensitive to sensor failures \cite{liu2018compare}. Symmetric nested arrays are robust to sensor failures but suffer from high mutual coupling \cite{liu2018optimizing}.     

\subsection{Problem Formulation}
To fully exploit the benefits arising from the difference coarray, the array design should satisfy Criteria 
\ref{criterion:closedform} to \ref{criterion:largearray}. However, the majority of existing array configurations do not meet them simultaneously. Furthermore, application-specific requirements should be considered, including properties such as symmetry, low mutual coupling, robustness to sensor failures, and more. 

For small scale, the design tasks are formulated as optimization problems which can be solved by brute force methods. However, these problems are NP-hard in general, hence, for moderate and large scale they become intractable. Therefore, an efficient scalable approach for array design is required.     

To address this issue, we propose a fractal array design in which an array, called a generator, is recursively enlarged based on its difference coarray. We study the proposed array with respect to Criteria \ref{criterion:closedform}-\ref{criterion:largearray} and other important properties. We show that the resultant fractal arrays enjoy the same properties as the generator. Any array can be used as a generator, thus, extending known sparse array configurations. Moreover, a small-scale generator with required properties can be designed and then expanded by the proposed scheme to construct large fractal arrays which share the same properties.

We emphasize that our goal is not to present a specific array configuration which is optimal or superior to previously proposed arrays in terms of a specific property. Our aim is to offer a flexible framework for constructing large sparse arrays which exhibit multiple properties with the appropriate balance between them, determined by the specific application.      

\section{Sparse Fractal Array Design}
\label{sec:fractal}
Fractal arrays possess an
inherent self-similarity in their geometrical structure and have been used over the years in the design of radiating  systems, allowing multiwavelength operation. However, so far, fractals have not been studied extensively in the context of sparse array design.

In this section we present our main approach to designing sparse fractal arrays with increased degrees of freedom. To that end, we first briefly describe well-known fractals called Cantor arrays which exhibit a relatively small number of DOF \cite{liu2017maximally}. Then, we introduce a simple array design in which we recursively expand a generator array in a fractal fashion, allowing to construct arbitrarily large arrays that satisfy Criteria \ref{criterion:closedform} to \ref{criterion:largearray}. In addition, the proposed scheme can be seen as a generalization of Cantor arrays, leading to sparse fractal arrays with increased DOF.  

For simplicity, we assume henceforth that the leftmost element of an arbitrary array is located at 0. Otherwise, the array can be translated to fulfill this assumption. 

\subsection{Cantor Arrays}
\label{subsec:cantor}
Given two integer sets $\bb{A}$ and $\bb{B}$, we define their \textit{sum set} as
\begin{equation*}
\bb{A}+\bb{B}\defeq\{a+b\;|\;a\in\bb{A},\,b\in\bb{B}\}.
\end{equation*}
Cantor arrays \cite{feder2013fractals,falconer2004fractal} are fractal arrays defined recursively as follows:
\begin{align}
\begin{split}
&\bb{C}_0\defeq\{0\}, \\
&\bb{C}_{r+1}\defeq\bb{C}_r\cup(\bb{C}_r+3^r),\;r\in\mds{N},
\end{split}
\label{eq:Cantor}
\end{align}
where $\cup$ denotes the union operator. Note that the array definition (\ref{eq:Cantor}) is equivalent to the definition given in \cite{liu2017maximally}.
Cantor arrays are symmetric and $\bb{C}_r$ has $N=2^r$ physical elements.
See examples of Cantor arrays in Fig.\,\ref{fig:cantor}.

\begin{figure}
 \centering
 \includegraphics[trim={3cm 7cm 3cm 6.5cm},clip,height = 2cm, width = 0.8\linewidth]{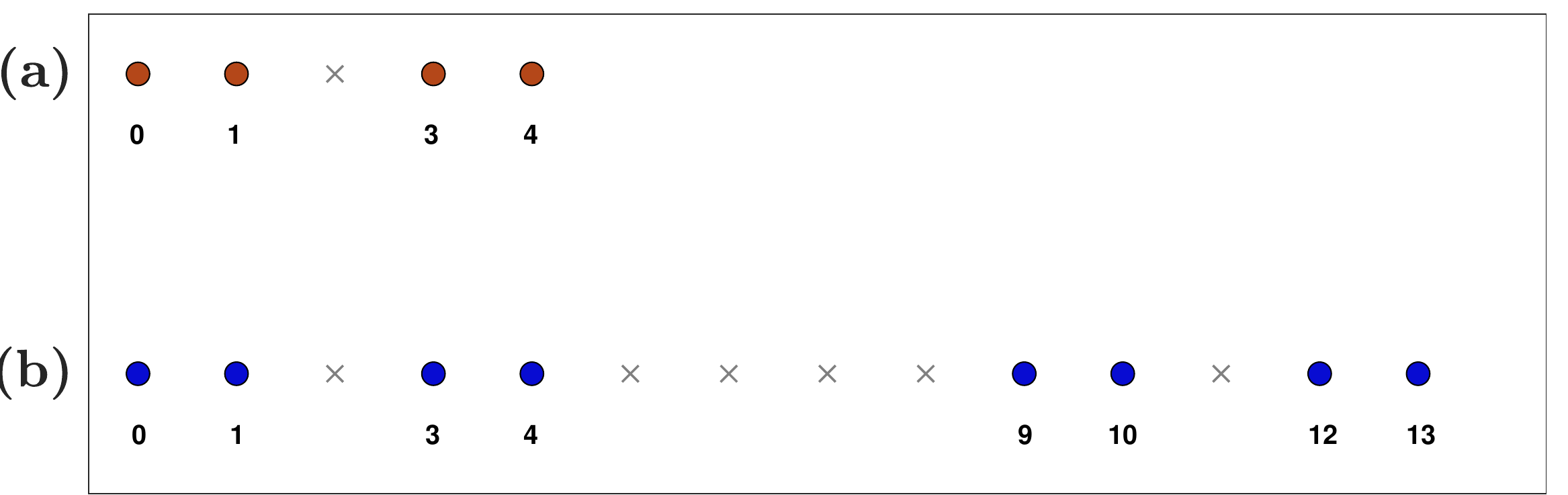}
\caption{Cantor arrays with (a) $r=2$ and (b) $r=3$.}
  \label{fig:cantor}
 \end{figure}

As proven in \cite{liu2017maximally}, $\bb{C}_r$ has a hole-free difference coarray $\bb{D}_r$ with size $\abs{\bb{D}_r}=3^r$. Hence, Cantor arrays satisfy Criteria  \ref{criterion:closedform} to \ref{criterion:largearray} along with symmetry. However,
\begin{equation*}
    \abs{\bb{D}_r}=3^r=(2^{\log_23})^r=(2^r)^{\log_23}=N^{\log_23}, 
\end{equation*}
i.e., the difference coarray has size $\mathcal{O}(N^{\log_23})\approx \mathcal{O}(N^{1.585})$. This asymptotic result is inferior to the performance
obtained by other sparse arrays such as nested arrays and MRAs that have size $\mathcal{O}(N^2)$ difference coarray. In addition, the number of sensors $N$ is constrained to be a power of two.

To overcome these limitations, we next present a fractal scheme in which the generator is taken to be a sparse array. We show that Cantor arrays are a  special case of the proposed arrays and that an appropriate choice of the generator leads to fractal arrays with increased DOF.  

\subsection{Sparse Fractal Arrays}
\label{subsec:fractal}
Here we introduce a fractal array design in which a generator array is expanded in a simple recursive fashion. We study the properties of the resultant arrays and prove they fulfill Criteria \ref{criterion:closedform} to \ref{criterion:largearray} when the generator satisfies them.     

Consider an $L$ element linear array whose sensor locations correspond to an integer set $\bb{G}\;(\abs{\bb{G}}=L)$. Denote the difference coarray of $\bb{G}$ by $\bb{D}$ and  let $\bb{U}$ be the central ULA of $\bb{D}$. We propose the following fractal array definition:
\begin{align}
\begin{split}
&\bb{F}_0 \defeq \{0\}, \\
&\bb{F}_{r+1} \defeq \bigcup_{n\in\bb{G}} (\bb{F}_r+n|\bb{U}|^r),\quad r\in\mds{N},
\end{split}
\label{eq:fractal}
\end{align}
where $|\bb{U}|$ denotes the cardinality of the set $\bb{U}$. Note that $\bb{F}_1$ is exactly the array $\bb{G}$, known as the \textit{generator} in fractal terminology \cite{werner1999fractal,werner2003overview,feder2013fractals,falconer2004fractal}. For brevity, we define the array translation factor as $M\defeq\abs{\bb{U}}$ and we refer to $r$ as the array order. Definition (\ref{eq:fractal}) is similar to the definition given in \cite{cohen2019fractal}, but here we do not assume the generator has a hole-free difference coarray and the translation factor is determined by the central ULA.     

The fractal array $\bb{F}_r$ consists of replicas of $\bb{F}_{r-1}$ translated according to the element locations of $\bb{G}$ and the cardinality of $\bb{U}$. This process is repeated a finite number of times, given by the array order $r$, to create a fractal array composed of copies of the generator where the number of sensors is $L^r$ at most. When $\bb{G}=[0\;1]$, the array definition (\ref{eq:fractal}) reduces to (\ref{eq:Cantor}), therefore, the proposed arrays can be seen as a generalization of Cantor arrays. Unlike previous related work \cite{puente1996fractal,werner1999fractal,werner2003overview,feder2013fractals,falconer2004fractal} where the translation factor can be an arbitrary natural number, here, it is directly determined by the generator's difference coarray. 
Fig.\,\ref{fig:fractals} illustrates fractal arrays created by (\ref{eq:fractal}) using different generators.

\begin{figure*}[ht]
 \centering
 \includegraphics[trim={0cm 5cm 0cm 4.5cm},clip,height = 6cm, width = 0.9\linewidth]{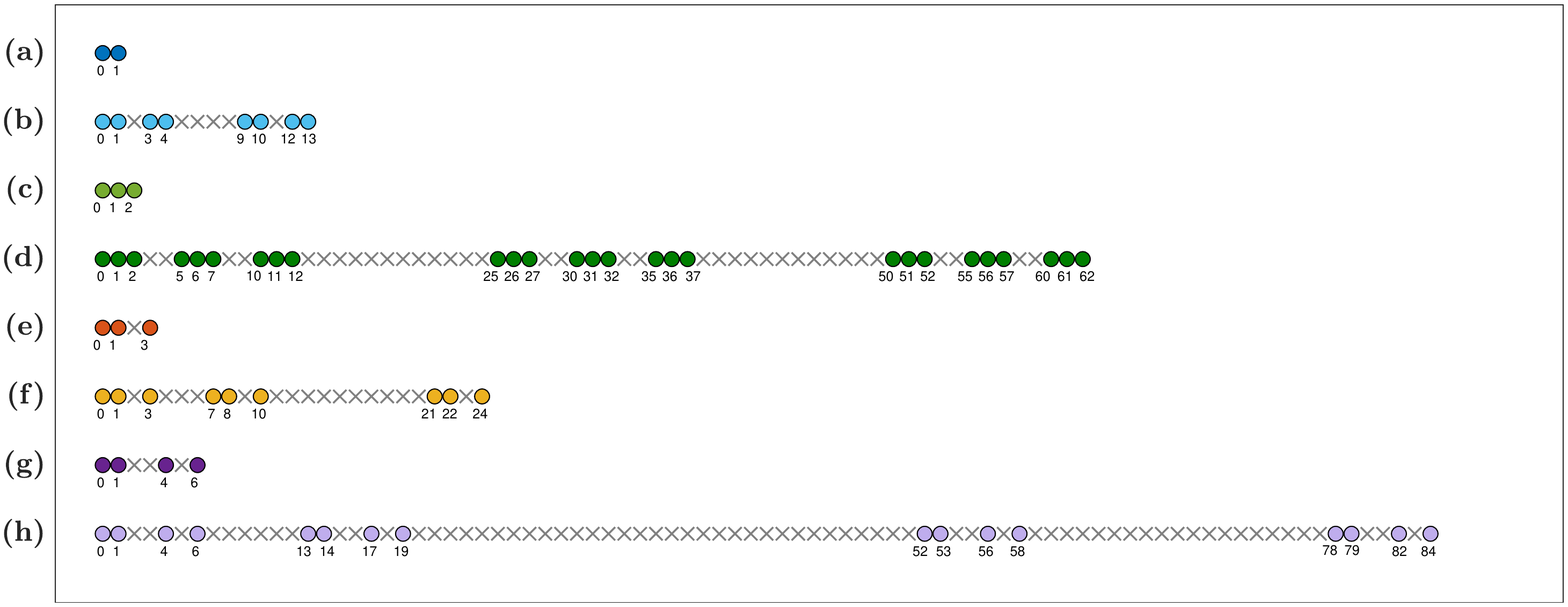}
\caption{\textbf{Fractal Arrays.} Different generator arrays (a) Cantor array, (c) ULA, (e) MRA, (g) MHA and their respective fractal extensions with (b) $r=3$, (d) $r=3$, (f) $r=2$, (h) $r=2$.}
  \label{fig:fractals}
 \end{figure*}

Next, we show that an appropriate choice of the generator leads to fractal arrays which meet Criteria \ref{criterion:closedform} to \ref{criterion:largearray}.     
First, the suggested arrays are expressed in closed-form, hence, they automatically satisfy Criterion\,\ref{criterion:closedform}. The result related to Criteria\,\ref{criterion:holefree} is stated in the following theorem.
 
\begin{theorem}
Consider an array $\bb{G}$ whose difference coarray is $\bb{D}$. Let $\bb{F}_r$ be the fractal array created according to (\ref{eq:fractal}) with $\bb{G}$ for some  fixed $r$. The difference coarray $\bb{D}_r$ of $\bb{F}_r$ is a hole-free array if $\bb{D}$ is hole-free. Moreover, we have 
\begin{equation*}
\bb{D}_r=\left[-\frac{M^r-1}{2},\frac{M^r-1}{2}\right],
\end{equation*}
where $M=\abs{\bb{U}}=\abs{\bb{D}}$.
\label{theorem:holefree}
\end{theorem}

\begin{proof}
We prove the theorem by induction.
\begin{itemize}
\item \textbf{Base} ($k=1$): In this case $\bb{F}_1=\bb{G}$. Hence,  $\bb{D}_1=\bb{D}$ which can be written as 
\begin{equation*}
\bb{D}_1=\bb{D}=\left[-\frac{M-1}{2},\frac{M-1}{2}\right],
\end{equation*}
where $M=\abs{\bb{D}}$ since $\bb{D}$ is hole free.
\item \textbf{Assumption} ($k=r$): $\bb{D}_r$ is a hole-free array given by 
\begin{equation*}
    \bb{D}_r=\left[-\frac{M^r-1}{2},\frac{M^r-1}{2}\right].
\end{equation*}
\item \textbf{Step} ($k=r+1$): 
By the definition of the difference coarray, we have
\begin{align*}
\bb{D}_{r+1}&\defeq \{k-l:\, k,l\in\bb{F}_{r+1}\} \\
&=\{s+uM^r-(t+vM^r):\, s,t\in\bb{F}_r,\,u,v\in\bb{G}\} \\
&=\{(s-t)+(u-v)M^r:\, s,t\in\bb{F}_r,\,u,v\in\bb{G}\} \\
&=\{m+nM^r:\, m\in\bb{D}_r,\,n\in\bb{D}\}.
\end{align*}
Since $\bb{D}$ and $\bb{D}_r$ are hole-free and satisfy $\abs{\bb{D}}=M,\,\abs{\bb{D}_r}=M^r$, we obtain that $\bb{D}_{r+1}$ consists of $M$ consecutive replicas of $\bb{D}_r$. Hence, $\bb{D}_{r+1}$ is hole-free and
\begin{equation*}
    \bb{D}_{r+1}=\left[-\frac{M^{r+1}-1}{2},\frac{M^{r+1}-1}{2}\right]. \qedhere
\end{equation*}
\end{itemize}
\end{proof}

Theorem\,\ref{theorem:holefree} demonstrates the importance of the difference coarray in our array definition. Choosing a generator whose difference coarray is hole-free, e.g. a ULA, leads to a fractal array that satisfies Criterion\,\ref{criterion:holefree} for any order $r$. For example, when the generator is a nested array the resultant fractal array has a contiguous difference coarray, while for a coprime generator array we get holes (see Fig.\,\ref{fig:holes}). Notice that when the generator has a hole-free difference coarray, it holds that $\abs{\bb{F}_r}=\abs{\bb{G}}^r$. 

We continue to the third criteria.

\begin{figure*}[ht]
 \centering
 \includegraphics[trim={0cm 1.5cm 0cm 2.5cm},clip,height = 8cm, width = 0.8\linewidth]{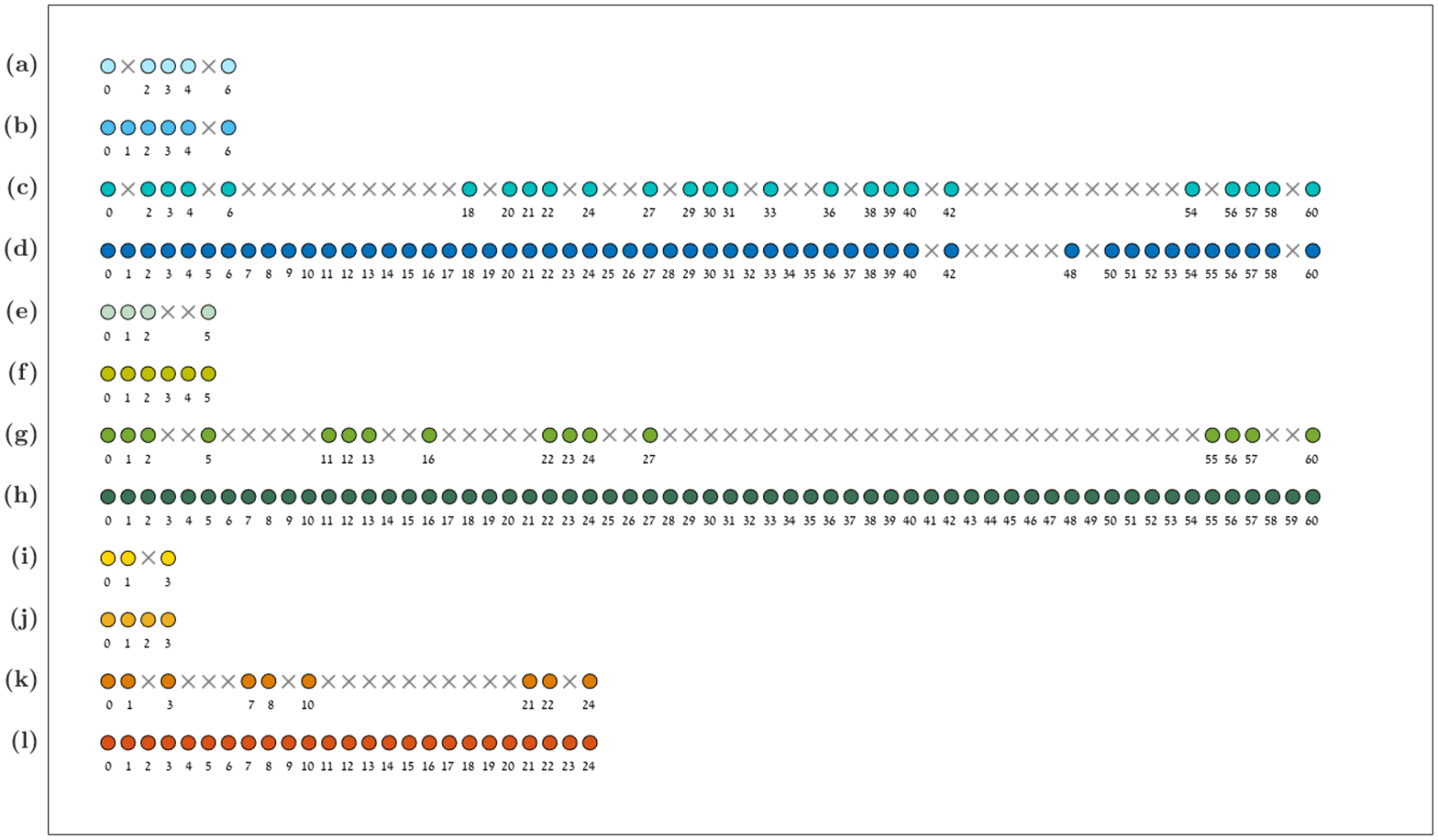}
\caption{\textbf{Difference Coarray.} Different generator arrays (a) expanded coprime array, (e) nested array, (i) MRA and the non-negative parts of their difference coarrays given by (b), (f) and (j) respectively. The fractal extensions for $r=2$ of the generators are shown in (c), (g) and (k), while the corresponding non-negative parts of their difference coarrays are (d), (h) and (l). }
  \label{fig:holes}
 \end{figure*}
 
\begin{theorem}
Consider an $L$ element array $\bb{G}$ whose difference coarray satisfies $M\defeq\abs{\bb{U}}=
\mathcal{O}(L^2)$. Let $\bb{F}_r$ be the fractal array created according to (\ref{eq:fractal}) with $\bb{G}$ for some  fixed $r$. Then, the difference coarray $\bb{D}_r$ of $\bb{F}_r$ satisfies
\begin{equation*}
\abs{\bb{D}_r}=\mathcal{O}(N^2),
\end{equation*}
where $N\leq L^r$ is the number of physical sensors in $\bb{F}_r$.
\label{theo:large}
\end{theorem}

\begin{proof}
Denote by $\bb{U}_r$ the central ULA of $\bb{D}_r$. We first prove by induction that
\begin{equation*}
    \left[-\frac{M^r-1}{2},\frac{M^r-1}{2}\right]\subseteq \bb{U}_r,
\end{equation*}
implying that $\abs{\bb{U}_r}=\mathcal{O}(M^r)$.
\begin{itemize}
\item \textbf{Base} ($k=1$): In this case $\bb{F}_1=\bb{G}$.  Therefore, $\bb{U}_1=\bb{U}$ which can be written as \begin{equation*}
\bb{U}_1=\bb{U}=\left[-\frac{M-1}{2},\frac{M-1}{2}\right],
\end{equation*}
where $M=\abs{\bb{U}}$ since $\bb{U}$ is a symmetric ULA by definition.
\item \textbf{Assumption} ($k=r$): Assume that
\begin{equation*}
    \left[-\frac{M^r-1}{2},\frac{M^r-1}{2}\right]\subseteq \bb{U}_r.
\end{equation*}
\item \textbf{Step} ($k=r+1$):
% As shown in the proof of Theorem\,\ref{theorem:holefree},
% \begin{equation*}
%     \bb{D}_{r+1}=\{m+nM^r:\, m\in\bb{D}_r,\,n\in\bb{D}\}.
% \end{equation*}
We define the following sets
\begin{align*}
    &\bb{T}_r\triangleq  \left[-\frac{M^r-1}{2},\frac{M^r-1}{2}\right], \\
    &\bb{Y}_r\triangleq \Big\{m+nM^r:\, m\in\bb{T}_r,\,n\in\bb{U}\Big\}, \\
    &\bb{V}_r\triangleq \{m+nM^r:\, m\in\bb{U}_r,\,n\in\bb{U}\}.
\end{align*}
Notice that we can write $\bb{Y}_r$ explicitly as 
\begin{equation*}
    \bb{Y}_r=\left[-\frac{M^{r+1}-1}{2},\frac{M^{r+1}-1}{2}\right].
\end{equation*}
Furthermore, it holds that $\bb{U}_r\subseteq \bb{V}_r$ and by definition of the central ULA we obtain that $\bb{U}_{r+1}$ includes $\bb{V}_r$, i.e., $\bb{V}_r\subseteq \bb{U}_{r+1}$. 
By the induction assumption, we have that $\bb{T}_r\subseteq\bb{U}_r$,  hence, $\bb{Y}_r\subseteq \bb{V}_r$ which in turn suggests that
\begin{equation*}
    \bb{Y}_r=\left[-\frac{M^{r+1}-1}{2},\frac{M^{r+1}-1}{2}\right]\subseteq \bb{U}_{r+1},
\end{equation*}
since $\bb{V}_r\subseteq \bb{U}_{r+1}$. Thus, we get that 
\begin{equation*}
    \abs{\bb{U}_{r+1}}\geq \Bigg|\left[-\frac{M^{r+1}-1}{2},\frac{M^{r+1}-1}{2}\right]\Bigg|=M^{r+1}.
\end{equation*}
\end{itemize}
Now, since $ \bb{U}_r \subseteq \bb{D}_r$ we have
\begin{equation*}
\abs{\bb{D}_r}\geq \abs{\bb{U}_r}\geq 2\left(\frac{M^r-1}{2}\right)+1=M^r,
\end{equation*}
i.e., $\abs{\bb{D}_r}=\mathcal{O}(M^r)$.
Finally, since $M=\mathcal{O}(L^{2})$ and $N\leq L^r$ we obtain that  $M^r=\mathcal{O}(L^{2r})$ and thus 
\begin{equation*}
\abs{\bb{D}_r}=\mathcal{O}(M^r)=\mathcal{O}(L^{2r})=\mathcal{O}(N^2),
\end{equation*}
completing the proof. \qedhere
\end{proof}

Theorem\,\ref{theo:large} implies that for a proper choice of the generator, e.g. a co-prime array, we obtain fractal arrays that fulfill Criterion\,\ref{criterion:largearray}. In particular, the size of the central ULA is $\abs{\bb{U}}=\mathcal{O}(N^2)$ where $N$ is the number of physical elements, as occurs in nested arrays and coprime arrays. Thus, the proposed fractal arrays are an improvement over Cantor arrays since they exhibit increased DOF and their number of sensors $N$ is not necessarily a power of two.

From the last two theorems, we can use a generator with a large contiguous difference coarray, such as a nested array (Fig.\,\ref{fig:holes}), to create an arbitrarily large array that satisfies Criteria\,\ref{criterion:closedform} to \ref{criterion:largearray}. In the following section, we show that similar results can be obtained for other significant array properties. 

\section{Sparse Fractal Array Properties}
\label{sec:properties}

As shown in the previous section, we can build fractal arrays that satisfy the major criteria of sparse array design. However, other known array geometries also meet these criteria, for instance, nested arrays. To emphasize the advantage of the proposed fractal arrays, we extend our study to other desired array properties \cite{liu2017maximally,svantesson1999direction,liu2016super,liu2018compare} which are important in diverse applications. Similar to Section\,\ref{subsec:fractal}, we show that these fractal array properties are induced by the generator.

\subsection{Symmetry}
Symmetric arrays are favorable in various applications ranging from DOA estimation \cite{ye2007doa} to ultrasound imaging \cite{cohen2018coba,cohen2018sparse,cohen2017sparse, chernyakova2018fourier}. Symmetry induces a special structure on the acquired signals which can be exploited to reduce the computational burden, aid in calibration and ultimately improve DOA estimation \cite{haupt1994thinned,friedlander1991direction,xu1992detection,xu2006deflation,ye2007doa,lin2006blind}. 
% Below, we provide the definition of array symmetry followed by our theorem on the symmetry of the discussed fractal arrays.    

\begin{definition}[Reversed Array]
Consider a sensor array $\bb{G}$. The
reversed version of an array $\bb{G}$ is defined as 
\begin{equation*}
\hat{\bb{G}}\defeq\{\max(\bb{G})+\min(\bb{G})-n\,|\;n\in\bb{G}\}.
\end{equation*}
As we assume that $\min(\bb{G})=0$ for any array, the above definition reduces to
\begin{equation*}
\hat{\bb{G}}=\{\max(\bb{G})-n\,|\;n\in\bb{G}\}.
\end{equation*}
\end{definition}

\begin{definition}[Symmetric Array]
Consider a sensor array $\bb{G}$ and denote by $\hat{\bb{G}}$ its reversed array. We say that an array $\bb{G}$ is symmetric if $\bb{G}=\hat{\bb{G}}$.
\end{definition}
\noindent The following theorem states a sufficient condition for fractal arrays to be symmetric.

\begin{theorem}
Let $\bb{F}_r$ be a fractal array whose generator is $\bb{G}$. 
Then, $\bb{F}_r$ is symmetric if $\bb{G}$ is symmetric.
\label{theo:summetry}
\end{theorem}

\begin{proof}
We prove the theorem by induction.
\begin{itemize}[leftmargin=*]
\item \textbf{Base} ($k=1$): $\bb{F}_1=\bb{G}$, hence, $\bb{F}_1$ is symmetric.
\item \textbf{Assumption} ($k=r$): $\bb{F}_r$ is symmetric.
\item \textbf{Step} ($k=r+1$): First, we assume $\min(\bb{G})=0$, leading to $\min(\bb{F}_r)=0$.
% \begin{equation*}
% \min(\bb{G})=0,\;\min(\bb{F}_r)=0.
% \end{equation*}
In addition, we can rewrite $\bb{F}_{r+1}$ as
\begin{equation*}
\bb{F}_{r+1}=\{m+nM^r:\; m\in\bb{F}_r,\,n\in\bb{G}\}.
\end{equation*}
Hence, we get
\begin{align}
\begin{split}
&\max(\bb{F}_{r+1})=\max(\bb{F}_r)+\max(\bb{G})M^r, \\
&\min(\bb{F}_{r+1})=\min(\bb{F}_r)+\min(\bb{G})M^r=0.
\end{split}
\end{align}
Let $\hat{\bb{G}}$, $\hat{\bb{F}}_r$ and $\hat{\bb{F}}_{r+1}$ denote the reversed arrays of $\bb{G}$, $\bb{F}_r$ and $\bb{F}_{r+1}$ respectively. From the above equations we obtain 
\begin{align*}
\hat{\bb{F}}_{r+1}&\defeq 
\{\max(\bb{F}_{r+1})+\min(\bb{F}_{r+1})-l:\, l\in \bb{F}_{r+1}\} \\
&=\{\max(\bb{F}_{r+1})-l:\, l\in \bb{F}_{r+1}\} \\
&=\{\max(\bb{F}_r)+\max(\bb{G})M^r-l:\,l\in \bb{F}_{r+1}\}.
\end{align*}
Note that each $l\in\bb{F}_{r+1}$ can be expressed as $l=m+nM^r$ for some $m\in \bb{F}_r$ and $n\in \bb{G}$. Therefore,
{\small
\begin{align*}
&\hat{\bb{F}}_{r+1}
=\{\max(\bb{F}_{r+1})-(m+nM^r):\, m\in \bb{F}_r,\, n\in \bb{G}\} \\
&=\{\max(\bb{F}_r)+\max(\bb{G})M^r-(m+nM^r):\, m\in \bb{F}_r,\, n\in \bb{G}\} \\
&=\{\max(\bb{F}_r)-m+\big(\max(\bb{G})-n\big)M^r:\, m\in \bb{F}_r,\, n\in \bb{G}\} \\
&=\{\hat{m}+\hat{n}M^r:\, \hat{m}\in \hat{\bb{F}}_r,\, \hat{n}\in \hat{\bb{G}}\},
\end{align*}}
\hspace{-10pt} where the last equality follows from the definition of the reversed array.
Since $\hat{\bb{G}}=\bb{G}$ and $\hat{\bb{F}}_r=\bb{F}_r$, we get 
\begin{equation*}
  \hat{\bb{F}}_{r+1}=\{\hat{m}+\hat{n}M^r:\, \hat{m}\in \bb{F}_r,\, \hat{n}\in \bb{G}\}=\bb{F}_{r+1}, 
\end{equation*}
so that $\bb{F}_{r+1}$ is symmetric. \qedhere
\end{itemize}
\end{proof}

\noindent For Cantor arrays, the generator is $\bb{G}=[0\;1]$ which is symmetric. Thus, Theorem\,\ref{theo:summetry} provides an alternative explanation for the result presented in \cite{liu2017maximally} regarding the symmetry of Cantor arrays.

\subsection{Weight Function and Beampattern}
Next, we consider the weight function and beampattern of fractal arrays. The weight function is an important characteristic of a linear array and is associated with several array properties such as mutual coupling\cite{liu2016super,liu2016superII}, array economy \cite{liu2017maximally} and robustness \cite{liu2018robust,liu2018compare,liu2018optimizing}. In addition, the array beampattern, related to the weight function through the Fourier transform, dictates the array directivity and impacts the performance of correlation-based estimators and beamformers.

We start with defining the weight function and the corresponding beampattern.

\begin{definition}[Weight Function]
Consider a sensor array $\bb{G}$. The weight function $w_\bb{G}(m)$ equals the number of sensor pairs in $\bb{G}$ with separation $m$. Namely, \cite{liu2017maximally}  
\begin{equation*}
w_\bb{G}(m)=\abs{\{(n_1,n_2)\in\bb{G}^2:\; n_1-n_2=m\}},
\end{equation*}
where we define $\bb{S}^2\triangleq \bb{S}\times \bb{S}$ for any set $\bb{S}$. 
\end{definition}
\noindent Note that $w_\bb{G}(m)>0$ for any $m\in\bb{D}$ and zero otherwise, where $\bb{D}$ is the difference coarray of $\bb{G}$. Thus, the weight function is directly related to the difference coarray as well as the beampattern defined next.

\begin{definition}[Beampattern]
Consider a sensor array $\bb{G}$ whose difference coarray is $\bb{D}$. The beampattern of $\bb{G}$ is defined as
\begin{equation*}
B_\bb{G}(\omega)\defeq\sum_{m\in\bb{D}} w_\bb{G}(m)\exp\left(-j\omega m\right),
\end{equation*}
where $w_\bb{G}(m)$ is the weight function of $\bb{G}$, $\omega=\pi\sin(\theta)$ and $j=\sqrt{-1}$ is the imaginary unit. 
\end{definition}
\noindent Since $w_\bb{G}(m)$ is an even function \cite{pal2012nested}, the beampattern $B_\bb{G}(\omega)$ is real-valued. Moreover, $w_\bb{G}(m)=0$ for any $m\notin\bb{D}$, and therefore
\begin{align*}
    B_\bb{G}(\omega)&=\sum_{m\in\bb{D}} w_\bb{G}(m)\exp\left(-j\omega m\right) \\
    &=\sum_{m=-\infty}^{\infty} w_\bb{G}(m)\exp\left(-j\omega m\right)=\mathcal{F}\{w_\bb{G}\}(\omega),
\end{align*}
where $\mathcal{F}\{\cdot\}$ represents the discrete-time Fourier transform.

To derive both the weight function and the beampattern, we use the next definition.
\begin{definition}[$\ell$-Expansion]
Consider a sensor array $\bb{G}$ whose weight function is $w_\bb{G}(m)$. 
Given a positive integer $\ell$, we define the $\ell$-expansion of $w_\bb{G}(m)$ as the function
\begin{equation*}
w_\bb{G}^{\uparrow \ell}(m)\defeq
\begin{cases}
w_\bb{G}(n), & m=n\ell,\\
0, & otherwise.
\end{cases}
\end{equation*}
In other words, we create $w_\bb{G}^{\uparrow \ell}$ by adding $\ell-1$ zeros between each pair of consecutive entries of $w_\bb{G}$.   
\end{definition}
\noindent The Fourier transform of $w_\bb{G}^{\uparrow \ell}$ is given by
\begin{align*}
    \mathcal{F}\{w_\bb{G}^{\uparrow \ell}\}(\omega)&=\sum_{m=-\infty}^{\infty} w_\bb{G}^{\uparrow \ell}(m)\exp\left(-j\omega m\right) \\
    &=\sum_{n=-\infty}^{\infty} w_\bb{G}(n)\exp\left(-j\omega n\ell\right) \\
    &=B_\bb{G}(\ell w).
\end{align*}
Equipped with the above definition, we provide closed-form expressions for the weight function and the beampattern of fractal arrays in the following theorem.

\begin{theorem}
Let $\bb{F}_r$ be a fractal array whose generator is $\bb{G}$. Denote the weight function and beampattern of $\bb{G}$ by $w_{\bb{G}}$ and $B_\bb{G}(\omega)$ respectively. The weight function  $w_r$ of $\bb{F}_r$ is then given by
\begin{align*}
    &w_0(m)=\delta(m), \\
    &w_r(m)=\Conv{r-1} w_{\bb{G}}^{\uparrow M^i}(m), \quad r\geq 1,
\end{align*}
where $\delta(\cdot)$ is the Kronecker delta function and $\mathop{\scalebox{1.5}{$\circledast$}}$ denotes multiple convolution operations. The beampattern $B_r(\omega)$ of $\bb{F}_r$ is given by
\begin{equation*}
B_r(\omega)=\prod_{i=0}^{r-1}B_\bb{G}\left(M^i\omega\right),
\end{equation*}
where $B_\bb{G}(\omega)$ is the beampattern of $\bb{G}$.
\label{theo:weight}
\end{theorem}

\begin{proof}
We first prove the expression for the weight function. Note that $\bb{F}_0=\bb{D}_0=\{0\}$, hence, $w_0(m)=\delta(m)$.
For $r\geq 1$, we prove the result by induction.
\begin{itemize}[leftmargin=*]
\item \textbf{Base} $(k=1)$: In this case $\bb{F}_1=\bb{G}$, and indeed we get
\begin{equation*}
    w_1= \Conv{0} w_{\bb{G}}^{\uparrow M^i}=w_{\bb{G}}^{\uparrow 1}=w_{\bb{G}}.
\end{equation*}
\item \textbf{Assumption} $(k=r)$: $w_r=\Conv{r-1} w_{\bb{G}}^{\uparrow M^i}$.
\newline
\item \textbf{Step} $(k=r+1)$: By definition, 
\begin{equation*}
    w_{r+1}(m)\defeq\big|\{(m_1,m_2)\in \bb{F}_{r+1}^2:\, m_1-m_2=m\}\big|.
\end{equation*}
Recall that each $v\in\bb{F}_{r+1}$ can be expressed as $v=n_1+n_2M^r$ for some $n_1\in \bb{F}_r$ and $n_2\in \bb{G}$. Therefore,
{\footnotesize
\begin{align*}
&w_{r+1}(m)=\big|\{(m_1,m_2)\in \bb{F}_{r+1}^2:\, m_1-m_2=m\}\big| \\
&=\big|\{(n_1,n_2,l_1,l_2)\in \bb{F}_r^2\times \bb{G}^2:\, n_1+l_1M^r-(n_2+l_2M^r)=m\}\big| \\
&=\big|\{(n_1,n_2,l_1,l_2)\in \bb{F}_r^2\times \bb{G}^2:\, n_1-n_2+(l_1-l_2)M^r=m\}\big|.
\end{align*}}
\noindent Now, for a fixed $l\in\bb{D}$, consider the product between the number of tuples $(l_1,l_2)$ and the number of tuples $(n_1,n_2)$ that satisfy $l_1-l_2=l$ and $n_1-n_2=m-l\cdot M^r$ respectively. Notice that computing the latter for all $l\in\bb{D}$ and summing the results equals the desired number of quadruples $(n_1,n_2,l_1,l_2)$. Thus, we can write

{\footnotesize
\begin{align*}
w_{r+1}(m)&=\sum_{l\in\bb{D}}\Big(\big|\{(l_1,l_2)\in \bb{G}^2:\, l_1-l_2=l\}\big|\cdot\\
&\qquad\quad\;\;\, \big|\{(n_1,n_2)\in \bb{F}_r^2:\, n_1-n_2=m-l\cdot M^r\}\big|\Big).
\end{align*}}
In addition, we have that
{\small
\begin{align*}
    &w_{\bb{G}}(l)=\big|\{(l_1,l_2)\in \bb{G}^2:\, l_1-l_2=l\}\big|,\\
    &w_r(m-l\cdot M^r)=\big|\{(n_1,n_2)\in \bb{F}_r^2:\, n_1-n_2=m-l\cdot M^r\}\big|.
\end{align*}}
Hence, we obtain
\begin{equation*}
    w_{r+1}(m)=\sum_{l\in\bb{D}}w_{\bb{G}}(l)w_r(m-l\cdot M^r).
\end{equation*}
Since $w_{\bb{G}}(l)=0$ for any $l\notin \bb{D}$, we have $ w_{\bb{G}}^{\uparrow M^r}(n)=0$ for any $n\neq l\cdot M^r\,(l\in\mds{Z})$ which leads to 
\begin{align*}
 w_{r+1}(m)&=\sum_{l\in\mds{Z}}w_{\bb{G}}(l)w_r(m-l\cdot M^r) \\
&=\sum_{n\in\mds{Z}}w_{\bb{G}}^{\uparrow M^r}(n)w_r(m-n)=\big\{w_{\bb{G}}^{\uparrow M^r}\circledast w_r\big\}(m).
\end{align*}
By our assumption on $w_r$ we get
\begin{equation*}
    w_{r+1}=w_{\bb{G}}^{\uparrow M^r}\circledast w_r=\Conv{r} w_{\bb{G}}^{\uparrow M^i}.
\end{equation*}
\end{itemize}
Finally, multiple convolution operations translate to products in the Fourier domain, leading to the given expression for the beampattern.  
\end{proof}

Theorem\,\ref{theo:weight} provides simple expressions for both the weight function and beampattern which facilitate their analysis and optimization. These expressions suggest that the choice of the generator has a significant impact on the beampattern of the resultant fractal array with respect to the side-lobe level, grating lobes and the main-lobe, where the latter is also directly affected by the array order $r$. 

To demonstrate the above, we present a simple example in Fig.\,\ref{fig:bp} where we use an MHA as a generator $\bb{G}=[0\;1\;4\;6]$ and its fractal extension $\bb{F}_2$.
The generator's weight function $w_\bb{G}(m)$ and its $\ell$-extension $w_{\bb{G}}^{\uparrow \ell}(m)$  are shown in Fig.\,\ref{fig:bp}(a) and Fig.\,\ref{fig:bp}(b) respectively where $\ell=13$ is determined by the difference coarray. The weight function $w_2(m)$ of $\bb{F}_2$, displayed in Fig.\,\ref{fig:bp}(c), results from the convolution (marked in red) of $w_\bb{G}(m)$ and $w_{\bb{G}}^{\uparrow \ell}(m)$, as stated by Theorem\,\ref{theo:weight}. In addition, we provide in Fig.\,\ref{fig:bp} the beampatterns $B_\bb{G}(\omega)$, $B_\bb{G}(\ell\omega)$ and $B_2(\omega)$ which correspond to the weight functions $w_\bb{G}(m)$, $w_{\bb{G}}^{\uparrow \ell}(m)$ and $w_2(m)$ respectively. As can be seen, the beampattern $B_\bb{G}(\ell\omega)$, given in Fig.\,\ref{fig:bp}(e), consists of $\ell$ copies of the generator's beampattern shown in Fig.\,\ref{fig:bp}(d), each compressed by a factor of $\ell$. We outline in blue the beampattern $B_2(\omega)$ of $\bb{F}_2$ in Fig.\,\ref{fig:bp}(f), which results from the product of $w_\bb{G}(m)$ and $B_\bb{G}(\ell\omega)$ shown in dashed-lines. 

Observing the beampatterns, we can infer that the main-lobe of $B_2(\omega)$ is determined by the main-lobe of $B_\bb{G}(\omega)$ (which in turn is dictated by the generator difference coarray) and the compression factor $\ell^{r-1}$ where $r$ is the array order. Therefore, the main-lobe width of $B_2(\omega)$ decreases as $\ell$ and $r$ increase, which is expected since the array aperture grows accordingly. Moreover, typically, the minimum inter-element spacing of the generator is half the wavelength to prevent grating lobes. This ensures that the beampattern of the resultant fractal array also exhibits no grating lobes thanks to the product expression given in Theorem\,\ref{theo:weight} which eliminates any grating lobes appearing in $B_\bb{G}(\ell\omega)$. Note, however, that the side-lobe level of $B_2(\omega)$ may be lower (as in Fig.\,\ref{fig:bp}) or higher than that of $B_\bb{G}(\omega)$. Therefore, the generator should be carefully designed to achieve adequate side-lobe levels. One can obtain low side-lobes by relying on the expression of the beampattern given in Theorem\,\ref{theo:weight} and computing appropriate weights to apply on the difference coarray of $\bb{F}_r$.

\begin{figure*}[ht]
 \centering
 \includegraphics[trim={0cm 4cm 0cm 4cm},clip,height = 7cm, width = 1\linewidth]{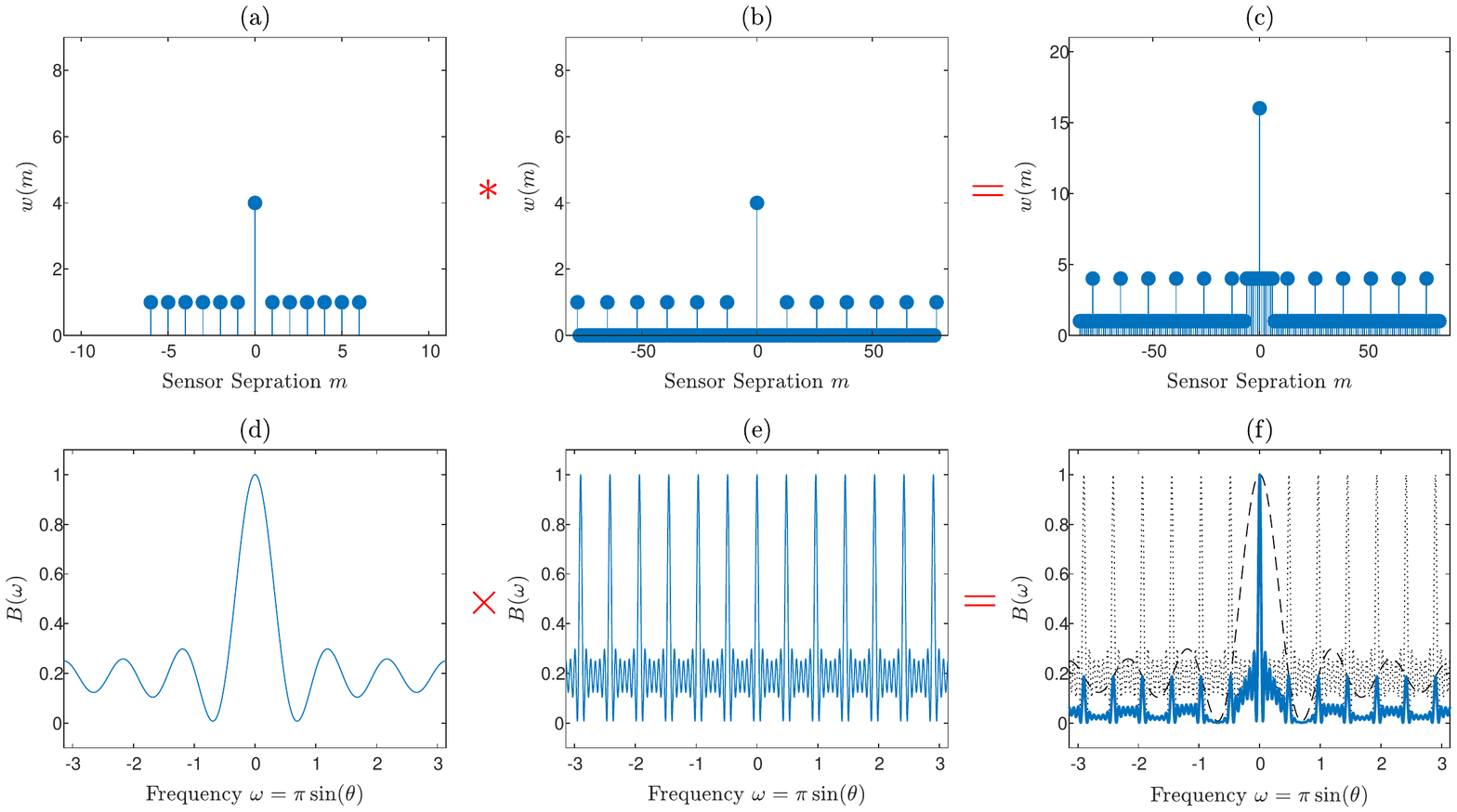}
\caption{\textbf{Weight Functions and Beampatterns.} (a) The weight function $w_\bb{G}(m)$ of $\bb{G}=[0\;1\;4\;6]$ and (b) its $\ell$-extension $w_{\bb{G}}^{\uparrow \ell}(m)$ with $\ell=13$. (c) The weight function $w_2(m)$ of the second-order fractal extension of $\bb{G}$, given by the convolution of (a) and (b) indicated by red symbols. The beampatterns related to $w_\bb{G}(m)$, $w_{\bb{G}}^{\uparrow \ell}(m)$ and $w_2(m)$ are shown in (d), (e) and (f) respectively, where the latter stems from the product marked in red.} 
  \label{fig:bp}
 \end{figure*}

\subsection{Array Economy and Robustness}
\label{subsec:robust}

Two contradicting properties of an array are the array economy and robustness to sensor failures.  
Array economy is related to the essentialness of each sensor which means that removing a specific sensor degrades the difference coarray \cite{liu2017maximally}. To reduce power and cost, one may desire to remove redundant array elements. When all sensors are essential, the array is said to be maximally economic \cite{liu2017maximally}. On the other hand, sensors might malfunction and create discontinuities (holes) in the difference coarray, making maximally economic arrays sensitive to element faults. Therefore, redundant elements may be added to make the array robust to sensor failures. 

Here we investigate fractal arrays with respect to the properties of essentialness and robustness using the notion of fragility introduced in \cite{liu2018robust}. To that end, we begin with the following definitions.

\begin{definition}[Essentialness]
Consider a sensor array $\bb{G}$ whose difference coarray is $\bb{D}$. Given a sensor $n\in\bb{G}$, define $\bb{G}_{-n}=\bb{G}\setminus \{n\}$ and denote the corresponding difference coarray by $\bb{D}_{-n}$. The sensor located at $n\in\bb{G}$ is said to be essential \cite{liu2017maximally} if $\bb{D}_{-n}\neq\bb{D}$. The sensor $n$ is inessetinal if it is not essential.  
\end{definition}

\begin{definition}[Maximally Economic]
A sensor array $\bb{G}$ is said to be \textit{maximally economic} if all of its sensors are essential.    
\end{definition}

\begin{lemma}
\cite{liu2017maximally} Consider a sensor array $\bb{G}$ whose weight function is $w_\bb{G}$.
If $n_1,n_2\in \bb{G}$ and $w_\bb{G}(n_1-n_2)= 1$, then $n_1$ and $n_2$ are
both essential with respect to $\bb{G}$. 
\label{lemma:essential}
\end{lemma}
\noindent Lemma\,\ref{lemma:essential} indicates that a sensor $n_1\in\bb{G}$ is essential if there exists $n_2\in\bb{G}$ such that $w_\bb{G}(n_1-n_2)= 1$. Note, however, that the converse may not be true, i.e., the lack of such $n_2$ does not automatically imply that $n_1$ is inessential. 

Given Lemma\,\ref{lemma:essential}, a sufficient but not necessary condition for an array $\bb{G}$ to be maximally economic can be defined as follows
\begin{equation}
    \tag{C1}
    \forall n_1\in\bb{G},\;\exists n_2\in\bb{G}:\;w_\bb{G}(n_1-n_2)= 1. 
    \label{condition:messa}
\end{equation}
This condition, however, requires to test the essentialness of each sensor, leading to heavy calculations for large arrays. The result of the next theorem avoids this computational burden by guaranteeing that fractal arrays satisfy condition (\ref{condition:messa}) when the generator fulfills it. 

\begin{theorem}
Let $\bb{F}_r$ be the fractal array generated from $\bb{G}$ whose difference coarray is hole-free. Then, $\bb{F}$ satisfies condition (\ref{condition:messa}) if $\bb{G}$ satisfies it.
\label{theo:economics}
\end{theorem}

\begin{proof}
First, observe that $\bb{F}_0=\{0\}$, hence, $w_0(0-0)=1$ and $\bb{F}_0$ is maximally economic and satisfies condition (\ref{condition:messa}). For $r\geq1$ we prove the theorem by induction.
\begin{itemize}[leftmargin=*]
\item \textbf{Base} $(k=1)$: In this case, $\bb{F}_1=\bb{G}$, therefore, $\bb{F}_1$ is maximally economic by satisfying condition (\ref{condition:messa}).
\item \textbf{Assumption} $(k=r)$: $\bb{F}_r$ is maximally economic by satisfying condition (\ref{condition:messa}).
\item \textbf{Step} $(k=r+1)$: Both $\bb{G}$ and $\bb{F}_r$ are  maximally economic by satisfying condition (\ref{condition:messa}). Hence, it holds that
\begin{align*}
&\forall l_1\in \bb{G},\; \exists l_2\in \bb{G}:\quad w_\bb{G}(l_1-l_2)=1, \\
&\forall n_1\in \bb{F}_r,\; \exists n_2\in \bb{F}_r:\quad w_r(n_1-n_2)=1.
\end{align*}

Consider an arbitrary $m_1\in\bb{F}_{r+1}$. By the array definition, there exist $n_1\in\bb{F}_r$ and $l_1\in\bb{G}$ such that $m_1=n_1+l_1M^r$. Moreover, since $\bb{G}$ and $\bb{F}_r$ are  maximally economic by satisfying condition (\ref{condition:messa}), there exist $ n_2\in\bb{F}_r$ and $ l_2\in\bb{G}$ such that  
\begin{equation*}
    w_\bb{G}(l_1-l_2)=1,\quad w_r(n_1-n_2)=1. 
\end{equation*}
Define $m=m_1-m_2$ where $m_2=n_2+l_2M^r$. Since $m_2\in\bb{F}_{r+1}$, we have that $m\in\bb{D}_{r+1}$ and $w_{r+1}(m)>0$.

We next prove that $w_{r+1}(m)=1$.
Following the proof of Theorem\,\ref{theo:weight} we write $w_{r+1}(m)$ as
\begin{equation*}
    w_{r+1}(m)=\sum_{l\in\bb{D}}w_{\bb{G}}(l)w_r(m-l\cdot M^r).
\end{equation*}
According to Theorem\,\ref{theorem:holefree}, $\bb{D}_r$ is hole-free and
\begin{equation*}
    \bb{D}_r=\left[-\frac{M^r-1}{2},\frac{M^r-1}{2}\right].
\end{equation*}
This implies that $w_r(n)=0$, for any $n\notin \left[-\frac{M^r-1}{2},\frac{M^r-1}{2}\right]$. Hence, for $l\in\bb{D}$ we have that $w_r(m-lM^r)>0$ if
\begin{align*}
    -\frac{M^r-1}{2}\leq m-lM^r \leq \frac{M^r-1}{2},
\end{align*}
i.e, when $l$ satisfies
\begin{align*}
        -\frac{M^r-1}{2}\leq n_1-n_2+(l_1-l_2)M^r-lM^r \leq \frac{M^r-1}{2}.
\end{align*}
From the latter we conclude that $w_r(m-lM^r)>0$ when $l=l_1-l_2$ and zero otherwise. This leads to
\begin{equation*}
    w_{r+1}(m) = w_\bb{G}(l_1-l_2)w_r(n_1-n_2)=1\cdot 1=1.
\end{equation*}
Therefore, $m_1$ is essential. Finally, since $m_1$ was chosen arbitrarily, all sensors in $\bb{F}_{r+1}$ are essential, i.e., $\bb{F}_{r+1}$ fulfills condition (\ref{condition:messa}) and it is maximally economic. \qedhere
\end{itemize}
\end{proof}

\noindent From Theorem\,\ref{theo:economics}, Cantor arrays are maximally economic since their generator $\bb{G}=[0\;1]$ is maximally economic and exhibits a hole-free difference coarray. Hence, Theorem\,\ref{theo:economics} extends the result of the economy of Cantor arrays presented in \cite{liu2017maximally}. 

Next, we study the robustness of fractal sparse arrays, starting with fragility.

\begin{definition}[Fragility]
Consider a sensor array $\bb{G}$. Define the following sub-array $\bb{E}=\{n\in\bb{G}\,|\;n\text{ is essential w.r.t }\bb{G}\}.$
The array fragility $F_\bb{G}$ is defined as \cite{liu2018robust}
\begin{equation*}
    F_\bb{G}\defeq \frac{\abs{\bb{E}}}{\abs{\bb{G}}}.
\end{equation*}
The fragility $F_\bb{G}$ quantifies the robustness/sensitivity of the difference coarray to sensor failures \cite{liu2018robust}.
\end{definition}

\noindent The fragility of any sparse array with $N\geq 4$ sensors satisfies $\frac{2}{N}\leq F_\bb{G}\leq 1$. For maximally economic sparse
arrays $\bb{E}=\bb{G}$, hence, $F_\bb{G}=1$. In contrast, an array such as a ULA and a RMRA \cite{liu2018optimizing} exhibit minimum fragility $F_\bb{G}=\frac{2}{N}$. 

The theorem below provides a relation between the fragility of the generator and the fragility of the fractal array created from it.

\begin{theorem}
Let $\bb{F}_r$ be the fractal array generated from $\bb{G}$ whose difference coarray $\bb{D}$ is hole-free. Denote by $F_\bb{G}$ and $F_r$ the fragility of $\bb{G}$ and $\bb{F}_r$ respectively. Then, it holds that
\begin{equation*}
    F_r\leq F_\bb{G},\quad \forall r\geq 1,
\end{equation*}
implying that $\bb{F}_r$ is at least as robust as $\bb{G}$ is.
\label{theo:robustness}
\end{theorem}

\begin{proof}
We prove the theorem by induction.
\begin{itemize}[leftmargin=*]
\item \textbf{Base} $(k=1)$: $\bb{F}_1=\bb{G}$, hence, $F_1=F_\bb{G}\leq F_\bb{G}$.
\item \textbf{Assumption} $(k=r):$ The fragility of $\bb{F}_r$ satisfies 
$F_r\leq F_\bb{G}$.
\item \textbf{Step} $(k=r+1)$: Denote by $L$ the number of elements in $\bb{G}$. Define $\bb{E}_r$ and $\bb{I}_r$ as the sets of essential and inessential sensors of $\bb{F}_r$ respectively. Notice that $\bb{E}_r\,\cap\,\bb{I}_r=\emptyset$, hence, $\abs{\bb{F}_r}=\abs{\bb{E}_r}+\abs{\bb{I}_r}$. The fragility of $\bb{F}_r$ can be written as
\begin{equation*}
    F_r=\frac{\abs{\bb{E}_r}}{\abs{\bb{F}_r}}=\frac{\abs{\bb{F}_r}-\abs{\bb{I}_r}}{\abs{\bb{F}_r}}=\frac{L^r-\abs{\bb{I}_r}}{L^r},
\end{equation*}
where $\abs{\bb{F}_r}=L^r$ since $\bb{G}$ has a hole-free difference coarray.
Consider an arbitrary $n\in\bb{I}_r$ and define $\bb{F}_r^{'}=\bb{F}_r\setminus\{n\}$. Denoting by $\bb{D}_r$ and $\bb{D}_r^{'}$ the difference coarrays of $\bb{F}_r$ and $\bb{F}_r^{'}$ respectively, we have that $\bb{D}_r^{'}=\bb{D}_r$. In addition, we define the following fractal extension of $\bb{F}_r^{'}$ as
\begin{equation*}
    \bb{F}_{r+1}^{'}=\{m+n\abs{\bb{D}}^r\,|\;m\in\bb{F}_r^{'},\,n\in\bb{G}\}.
\end{equation*}
Notice that $\bb{F}_{r+1}^{'}\subseteq \bb{F}_{r+1}$ and following the proof of Theorem\,\ref{theorem:holefree}, we get that $\bb{D}_{r+1}^{'}=\bb{D}_{r+1}$. Since the latter is true for any $n\in\bb{I}_r$, we have that $\abs{\bb{I}_{r+1}}\geq L\cdot\abs{\bb{I}_r}$ where $\bb{I}_{r+1}$ is the set of inessential sensors w.r.t $\bb{F}_{r+1}$. Therefore, it holds that
\begin{align*}
    F_{r+1}&=\frac{L^{r+1}-\abs{\bb{I}_{r+1}}}{L^{r+1}}\\
    &\leq\frac{L^{r+1}-L\cdot\abs{\bb{I}_r}}{L^{r+1}}=\frac{L^r-\abs{\bb{I}_r}}{L^r}=F_r\leq F_\bb{G},
\end{align*}
completing the proof. \qedhere
\end{itemize}
\end{proof}

Theorem\,\ref{theo:robustness} suggests a simple way for constructing large robust arrays. We demonstrate this approach in Fig.\,\ref{fig:robust} using fractal extensions of several MRAs and RMRAs, exemplifying a maximally economic fractal array versus a robust fractal array. 

As can be expected, the increase in the array robustness as the array order $r$ grows is at the expense of sensor redundancy, leading to lower DOF with respect to the number of physical sensors. Thus, the array order $r$ should be kept small, in the range 2-4 which is typically adequate for creating sufficiently large arrays. Moreover, the generator should be carefully chosen according to the quality and cost of the sensing device. For example, when the total budget or hardware constraints dictate the use of sensors susceptible to failures, the generator must be designed to exhibit low fragility while compromising on the size of the difference coarray with regard to the physical elements.  

\begin{figure*}[ht]
 \centering
 \includegraphics[trim={0cm 6.5cm 0cm 6.5cm},clip,height = 4.5cm, width = 0.9\linewidth]{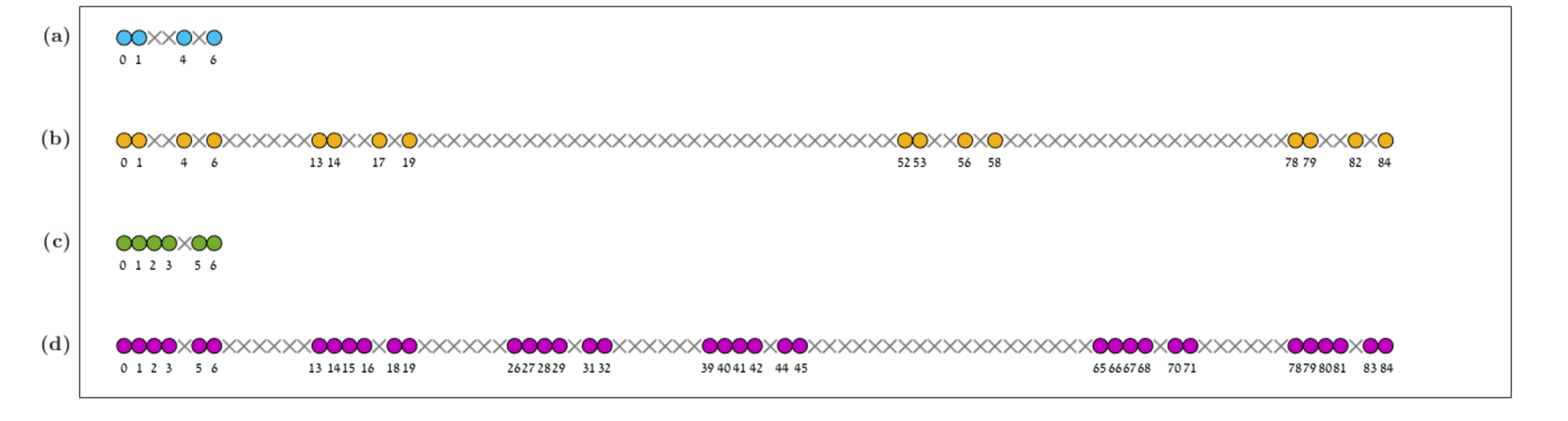}
\caption{\textbf{Economy and Robustness.} (a) Maximally economic array and (c) RMRA. The fractal extensions for $r=2$ of (a) and (c) are (b) and (d) respectively. The fragility values of arrays (a)-(d) are  1, 1, 1/3 and 1/9.}
  \label{fig:robust}
 \end{figure*}

\subsection{Mutual Coupling}
In Section\,\ref{subsec:model} we presented the signal model under the assumption that the elements do not interfere with each other. However, in practice, any sensor output is influenced by its adjacent sensors. This phenomena, called mutual coupling, has an adverse effect on the beampattern, degrading the performance of correlation-based estimators.

To address the effect of mutul coupling, we modify the signal model (\ref{eq:model}) as follows:
\begin{equation}
\label{eq:modelmc}
{\bf x}=\sum_{i=1}^K s_i{\bf Ca}(\theta_i)+{\bf w}={\bf CAs}+{\bf w},
\end{equation}
where $\bf C$ is a mutual coupling matrix derived from electromagnetics \cite{svantesson1999direction,liu2016super}. 
Assuming an $N$-element array $\bb{G}$, we consider a reduced coupling model
\cite{svantesson1999direction,svantesson1999modeling}
% change 2 \cite{svantesson1999direction,svantesson1999modeling,svantesson2000mutual,hui2007decoupling,pasala1994mutual}
where $\bf C$ is an $N\times N$ symmetric Toeplitz matrix given by
\begin{equation}
    {\bf C}=
    \begin{bmatrix}
    1   & c_i & \dots & c_q & 0 & \dots & 0 \\
    c_i & 1   &       &     & \ddots &       & \vdots \\
    \vdots &    &  \ddots     &     &  &    \ddots   & 0 \\
    c_q &    &      &  \ddots   &   &       & c_q \\
    0 & \ddots   &      &     &  \ddots  &       & \vdots \\
    \vdots &    &  \ddots    &     &    &  \ddots  & c_i \\
    0 & \dots   &  0    &  c_q  &  \dots  & c_i & 1 
\end{bmatrix},
\end{equation}
and $c_{|n-m|}\in\mds{C}$ represents the coupling coefficient of a pair of sensors $n,m\in\bb{G}$. The coefficients depend only on the element separation, leading to a coupling matrix with constant diagonals. Furthermore, they satisfy $c_0=1$ and $\abs{c_j}< \abs{c_i}$ for any $i,j\in\bb{D}$ such that $\abs{i}<\abs{j}$ where $\bb{D}$ is the difference coarray of $\bb{G}$. The coupling limit, represented by $q$, implies that for $i>q$ the coefficient $c_i$ can be neglected ($\abs{c_i}\approx 0$). Note that $q$ is a function of the number of sensors and the sensor separation distance. Here we assume that $q<\max(\bb{G})$ \cite{svantesson1999direction,liu2016super}.

When $\bf C$ is diagonal, the sensors are not coupled with each other. Therefore, the energy of the off-diagonal components of $\bf C$ are used to quantify the mutual coupling as defined below. 
\begin{definition}[Coupling Leakage \cite{liu2016super}]
Consider a sensor array $\bb{G}$ with a mutual coupling matrix $\bf C$.
We define the coupling leakage as
\begin{equation*}
    \mathcal{L}\defeq \frac{\norm{{\bf C}-\diag({\bf C})}_F}{\norm{\bf C}_F}
\end{equation*}
where $\norm{\cdot}_F$ denotes the Frobenius norm and $\diag({\bf C})$ is a matrix constructed by taking $\bf C$ and zeroing the off-diagonal elements.
\end{definition}
\noindent Note that $0\leq \mathcal{L}\leq 1$ and the smaller $\mathcal{L}$ is, the lower the mutual coupling. Under mild conditions, the proposed fractal arrays and their generators have the same coupling leakage, as shown in the following theorem.

\begin{theorem}
\label{theo:coupling}
Let $\bb{F}_r$ be the fractal array generated from $\bb{G}$. Denote by $\mathcal{L}_\bb{G}$ and $\mathcal{L}_r$ the coupling leakage of $\bb{G}$ and $\bb{F}_r$ respectively. Assuming the coupling limit $q$ satisfies $q<\max(\bb{G})$ and $q+\max(\bb{G})<\abs{\bb{U}}$, it holds that
\begin{equation*}
    \mathcal{L}_r=\mathcal{L}_\bb{G}.
\end{equation*}
\end{theorem}

\begin{proof}
First, for any $N\times N$ matrix $\bf Q$ we have that
\begin{equation*}
    \norm{ \textbf{I}_n\otimes \textbf{Q}}_F=\sqrt{n}\norm{\bf Q}_F,\quad 
    \diag(\textbf{I}_n\otimes \textbf{Q})=\textbf{I}_n\otimes \diag(\textbf{Q}),
\end{equation*}
where $\otimes$ represents the Kronecker product and $\textbf{I}_n$ is the $n\times n$ identity matrix for some $n\in\mds{N}$.

Under the assumptions $q<\max(\bb{G})$ and $q+\max(\bb{G})<\abs{\bb{U}}$, the fractal array $\bb{F}_r$ consists of $\abs{\bb{G}}^{r-1}$ non-overlapping replicas of $\bb{G}$ where each pair of copies are separated by more than $q$. Therefore, sensors from different replicas do not interfere with each other. This induces a block diagonal coupling matrix
\begin{equation}
    {\bf C}_r=
    \begin{bmatrix}
    {\bf C} & {\bf 0} & \dots & \dots & {\bf 0} \\
    {\bf 0} & \ddots & \ddots & & \vdots \\
    \vdots  & \ddots  & {\bf C} & \ddots & \vdots \\
    \vdots &  & \ddots & \ddots &  {\bf 0} \\
    {\bf 0} & \dots & \dots & {\bf 0} & {\bf C}
    \end{bmatrix},
\end{equation}
where $\bf C$ and ${\bf C}_r$ are the coupling matrices of $\bb{G}$ and $\bb{F}_r$ respectively. This relation can be expressed analytically as 
\begin{equation*}
    {\bf C}_r=\textbf{I}_{\tilde{r}}\otimes \textbf{C},
\end{equation*}
where $\tilde{r}\defeq \abs{\bb{G}}^{r-1}$.
Hence, the coupling leakage of $\bb{F}_r$ is
\begin{align*}
    \mathcal{L}_r&\defeq \frac{\norm{{\textbf{C}_r}-\diag({ \textbf{C}_r})}_F}{\norm{\textbf{C}_r}_F} \\
    &=\frac{\norm{{ \textbf{I}_{\tilde{r}} \otimes \textbf{C}}-\diag({\textbf{I}_{\tilde{r}} \otimes \textbf{C}})}_F}{\norm{\textbf{I}_{\tilde{r}} \otimes \textbf{C}}_F} \\
    &=\frac{\norm{{\textbf{I}_{\tilde{r}} \otimes \textbf{C}}-{\textbf{I}_{\tilde{r}} \otimes \diag(\textbf{C})}}_F}{\norm{\textbf{I}_{\tilde{r}} \otimes \textbf{C}}_F} \\
    &=\frac{\norm{\textbf{I}_{\tilde{r}} \otimes \big(\textbf{C}-\diag(\textbf{C})\big)}_F}{\norm{\textbf{I}_{\tilde{r}} \otimes \textbf{C}}_F} \\
    &=\frac{\sqrt{\tilde{r}}}{\sqrt{\tilde{r}}}\frac{\norm{\bf C- \diag(C)}_F}{\norm{\bf C}_F}=\mathcal{L}_\bb{G},
\end{align*}
completing the proof.
\end{proof}

 Unlike previous works, e.g. \cite{svantesson1999direction}, that assume the coupling limit satisfies $q<N$ for an $N$ element ULA, here, we require $q<\max(\bb{G})$ which is a weaker condition, since the number of sensors satisfies  $N\leq\max(\bb{G})$ for non-uniform arrays. Furthermore, $2\cdot\max(\bb{G})<\abs{\bb{U}}$ for most of the known sparse geometries such as coprime arrays and nested arrays, and in particular for any array whose difference coarray is hole free. Therefore, given that $q<\max(\bb{G})$, the majority of existing sparse arrays meet the second assumption in Theorem\,\ref{theo:coupling} of $q+\max(\bb{G})<\abs{\bb{U}}$. 

The result of Theorem\,\ref{theo:coupling} can be used to easily design large sparse arrays with predetermined coupling leakage. To demonstrate this, we use super-nested arrays and present their fractal extension in Fig.\,\ref{fig:coupling}. The coupling coefficients are chosen as $c_1=0.3e^{j\pi/3}$ and $c_i=\frac{c_1}{i}e^{-j(i-1)\pi/8}$ for $2\leq i\leq q$ where $q=\max(\bb{G})-1$.

\begin{figure*}[ht]
 \centering
 \includegraphics[trim={3cm 6.5cm 0cm 6.5cm},clip,height = 4cm, width = 0.8\linewidth]{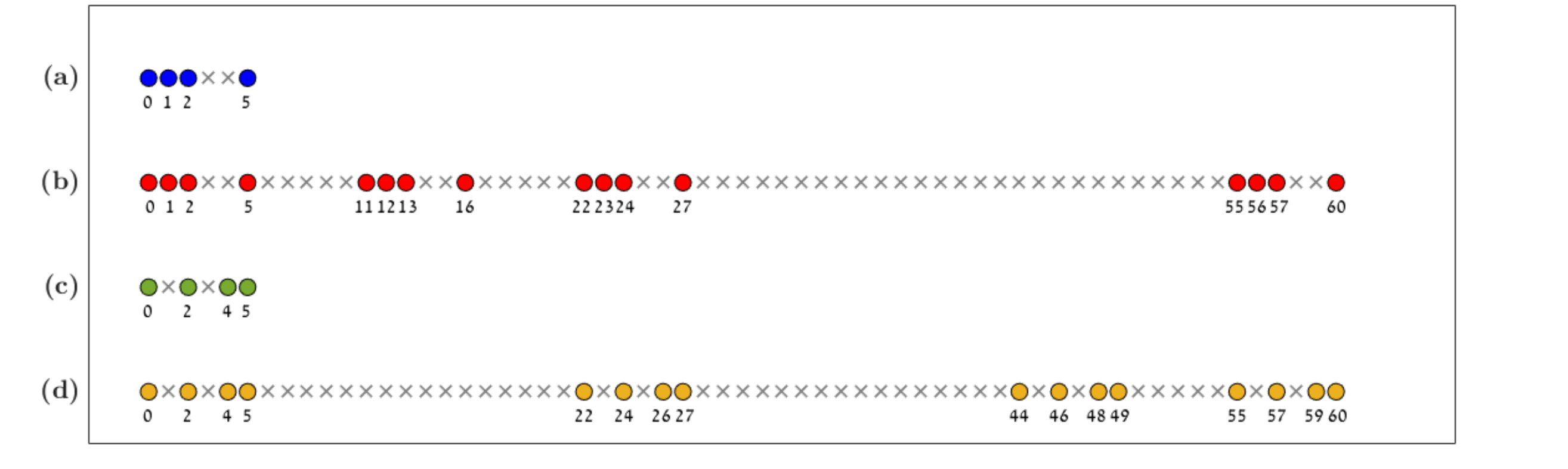}
\caption{\textbf{Mutual Coupling.} (a) A nested array followed by (b) its fractal extension for $r=2$, (c) a super nested array and (d) its fractal extension for $r=2$. The coupling leakage of (a) and (b) is 0.3159, whereas (c) and (d) achieve a lower value of 0.2676.}
  \label{fig:coupling}
 \end{figure*}
 
% \section{Extensions}

% \subsection{High-Order Statistics}
% \begin{definition}[High-Order Difference Coarray]
% \label{def:highdiff}
% Consider a sensor array $\bb{G}$. The $k$th order difference set of $\bb{G}$ is given by 
% \begin{equation*}
% \bb{D}^k\defeq\{\pm n_1\pm n_2\pm\dots\pm n_k\,|\;n_1,n_2,\cdots ,n_k\in\bb{G}\}.
% \end{equation*}
% The $k$th \textit{difference coarray} of $\bb{G}$ is the linear array $\bb{D}^k$.
% \end{definition} 

\vspace{-10 pt}
\subsection{Multi-Generators}
\label{subsec:multi}
Thus far, we have shown the benefits of sparse fractal arrays. However, a possible drawback of the proposed arrays is the exponential growth of the number of sensors with the array order. To circumvent this limitation, we extend (\ref{eq:fractal}) to the following array definition
\begin{align}
\begin{split}
&\bb{M}_0 \defeq \{0\}, \\
&\bb{M}_{r+1} \defeq \bigcup_{n\in\bb{G}_{r+1}} \Big(\bb{M}_r+n\prod_{i=0}^r\abs{\bb{U}_i}\Big),\quad r\in\{0,1,,\dots,R-1\},
\end{split}
\label{eq:multi}
\end{align}
where $\{\bb{G}_r\}_{r=1}^R$ are given generator arrays for a fixed $R>0$. To the best of our knowledge, the use of multiple generators has not been investigated before. In this scheme, a different generator is used at each iteration and the translation factor is determined by the difference coarrays of the generators from previous iterations. When all the generators are identical, the array (\ref{eq:multi}) reduces to (\ref{eq:fractal}), thus it generalizes the latter. Furthermore, it allows the number of sensors to be any composite number, not necessarily a perfect power, which grows gradually with the array order. However, these advantages come at the expense of designing multiple arrays, as each one of the generators may impact the resultant array. Moreover, depending on the chosen generators, the arrays created recursively may not exhibit self-similarity, i.e., they might not be fractal in practice.

In the following we present extensions of Theorem\,\ref{theorem:holefree} and Theorem\,\ref{theo:large} 
for the array configuration (\ref{eq:multi}). Theorem\,\ref{theo:multifree} describes the conditions for which the resultant fractal arrays have hole-free difference arrays, while Theorem\,\ref{theo:multilarge} relates to the size of the difference coarray and the number of DOF. The theorems presented before in regard to other properties, can be generalized in a similar fashion.  

\begin{theorem}
Let $R$ be a fixed positive integer and consider a series of generators $\{\bb{G}_i\}_{i=1}^R$ and their corresponding difference coarrays $\{\bb{D}_{\bb{G}_i}\}_{i=1}^R$. We assume $\bb{D}_{\bb{G}_i}$ is hole-free for any $1\leq i\leq R$.
Let $\bb{F}_r$ be the fractal array created according to (\ref{eq:multi}) with $\{\bb{G}_i\}_{i=1}^R$ for some fixed $r\leq R$. The, the difference coarray $\bb{D}_r$ of $\bb{M}_r$ is hole-free and we have
\begin{equation*}
\bb{D}_r=\left[-\frac{M_r-1}{2},\frac{M_r-1}{2}\right],
\end{equation*}
where $M_r=\prod\limits_{i=1}^r\abs{\bb{D}_{\bb{G}_i}}$ for all $1\leq r\leq R$.
\label{theo:multifree}
\end{theorem}

\begin{proof}
See Appendix\,\ref{app:mutlifree}.
\end{proof}

\begin{theorem}
Let $R$ be a fixed positive integer and consider a series of generators $\{\bb{G}_i\}_{i=1}^R$. We denote by $\{\bb{D}_{\bb{G}_i}\}_{i=1}^R$ the corresponding difference coarrays and their central ULAs by $\{\bb{U}_{\bb{G}_i}\}_{i=1}^R$. Let $\bb{F}_r$ be the fractal array created according to (\ref{eq:multi}) with $\{\bb{G}_i\}_{i=1}^R$ for some fixed $r\leq R$. Assuming $\abs{\bb{U}_{\bb{G}_i}}=\mathcal{O}(\abs{\bb{G}_i}^2)$ for all $1\leq i\leq r$, the difference coarray $\bb{D}_r$ of $\bb{M}_r$ satisfies
\begin{equation*}
\abs{\bb{D}_r}=\mathcal{O}(N^2)
\end{equation*}
for all $1\leq r\leq R$, where $N\leq\prod\limits_{i=1}^r \abs{\bb{G}}_i$ is the number of physical sensors in $\bb{M}_r$.
\label{theo:multilarge}
\end{theorem}

\begin{proof}
See Appendix\,\ref{app:multilarge}.
\end{proof}

Theorems \ref{theo:multifree} and \ref{theo:multilarge} show the generalized arrays (\ref{eq:multi}) fulfill Criteria\,\ref{criterion:holefree} and \ref{criterion:largearray}. The major benefit of (\ref{eq:multi}) is that it allows to combine diverse sparse geometries. The generators and their order need to be designed appropriately, since they affect the properties of the resultant fractal arrays, as shown in Fig.\,\ref{fig:multi}. For example, it can be verified that the fractal arrays are symmetric when all the generators are symmetric.   

\begin{figure*}[ht]
 \centering
 \includegraphics[trim={0cm 6.5cm 0cm 6.5cm},clip,height = 5cm, width = 1\linewidth]{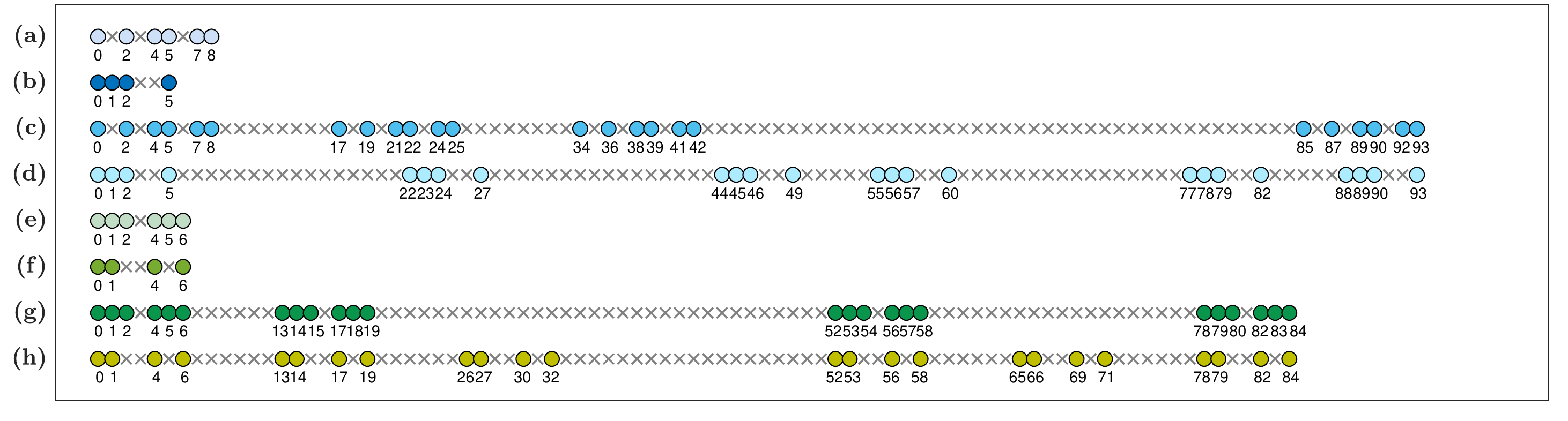}
\caption{\textbf{Multi-Generators.} (a) Super nested array, (b) nested array, (c) the fractal composition of (a) and (b), (d) the fractal composition of (b) and (a), (e) RMRA, (f) maximally economic array, (g) the fractal composition of (e) and (f), (h) the fractal composition of (f) and (a). The coupling leakage of (a)-(d) is 0.3016, 0.3159, 0.3016 and 0.3168 respectively. The fragility of (e)-(d) is 1/3, 1, 1/3 and 1/3 respectively.    }
  \label{fig:multi}
 \end{figure*}
 
%  \begin{table*}
%   \begin{center}
%     \begin{tabular}{|c|c|c|c|c|c|}
%       \hline % <-- Toprule here
%       \rowcolor{LightBlue}
%       \textbf{Array Geometry} & \textbf{Number of Sensors} & \textbf{Symmetric} & \textbf{Hole Free Difference Coarray} & \textbf{Fragility} & \textbf{Coupling Leakage} \\
%       \hline % <-- Midrule here
%       $\bb{G}=\left[0\;1\;2\;4\;7\;10\;13\;16\;18\;19\;20\right]$ & 11     & $\checked$ & $\checked$ & 0.4     & 0.30522 \\
%       $\bb{F}_2$                                      & 121    & $\checked$ & $\checked$ & 0.16   & 0.30522 \\
%       $\bb{F}_3$                                      & 1331   & $\checked$ & $\checked$ & 0.064  & 0.30522 \\
%       \hline % <-- Bottomrule here
%     \end{tabular}
%      \caption{\textbf{Array Analysis.} The first three arrays from the top are the soluions of the exhaustive search described in Section\,\ref{sec:experiments}. The following four arrays are the fractal expansions of the first array for $r=2,3,4,5$.}.
%      \label{tab:optimal}
%   \end{center}
% \end{table*}

\section{Numerical Experiments}
\label{sec:experiments}
Here we demonstrate the benefits of the proposed fractal scheme in designing large sparse arrays with multiple properties. We provide an analysis of the fractal arrays in comparison to several well-known sparse arrays mentioned earlier. 
Throughout the experiments, we assume that array motion \cite{qin2019doa}, virtual array interpolation \cite{zhou2018direction} and decoupling methods are not involved. 

We consider a representative array design with the following requirements: 
\begin{itemize}
    \item[(R1)] Symmetric array,
    \item[(R2)] Contiguous difference coarray (hole-free),
    \item[(R3)] Large difference coarray,
    \item[(R4)] Robustness to sensor failures -  $F\leq 0.3$,
    \item[(R5)] Mutual coupling - $\mathcal{L}\leq 1/3$,
    \item[(R6)] Constrained Aperture - $A \leq 840d$,
\end{itemize}
where $F$ is the array fragility, $\mathcal{L}$ is the array coupling leakage and $A$ denotes the size of the physical array aperture. We assume for simplicity that the element spacing is $d=\frac{\lambda}{2}=1$. For the coupling coefficients, we first parametrize $\abs{c_1}$ and then determine $\abs{c_2}$,...,$\abs{c_q}$ where $q=15$ assuming that the magnitudes of the coefficients are inversely proportional to the sensor separation $\left(\frac{\abs{c_j}}{\abs{c_i}}=\frac{i}{j}\right)$. The phases of the coefficients are drawn uniformly at random from $[-\pi,\,\pi)$.
Note that, in general, we may include in the design some conditions on the weight function or the beampattern.

Our task is to construct a large sparse array which fulfills the above requirements. A direct approach is to choose one out of the many state-of-the-art sparse configurations shown before such as MISC and RECA. However, these arrays do not meet the required specifications. To see this, we provide in Fig.\,\ref{fig:arraycomparison} a comparison between various known sparse arrays where each point corresponds to a different array, spatially positioned according to the array coupling leakage and fragility. As clearly seen, all of the aforementioned sparse arrays are out of the feasible region marked in green and determined by (R4)-(R5). Moreover, most of the discussed arrays exhibit low mutual coupling and high fragility. The reason for that lies in the fact that the design of these arrays focuses in redistributing the elements and sparsifying the array to obtain high DOF and low mutual coupling, which at the same time reduces their robustness. Hence, the challenge is to build a symmetric array which is relatively robust to sensor failures and exhibits low mutual coupling leakage.

\begin{figure*}[ht]
 \centering
 \includegraphics[trim={2.8cm 3.5cm 2cm 4cm},clip,height = 6.7cm, width = 0.8\linewidth]{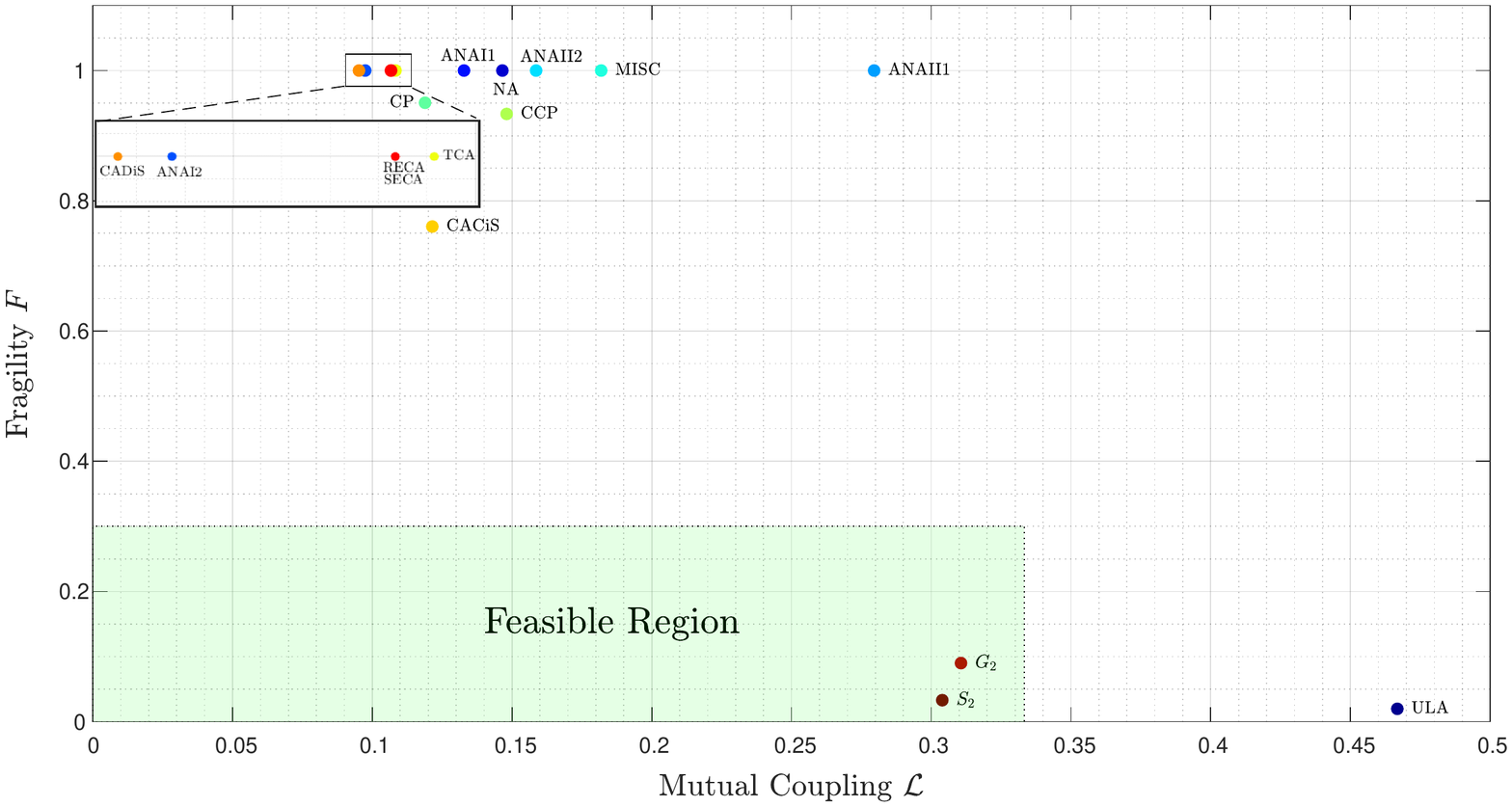}
\caption{\textbf{Array Comparison.} Displaying the tradeoff between coupling leakage and array fragility for various sparse configurations. The area colored in green marks the feasible region in which the mutual coupling is less than 1/3 and the fragility is below 0.3. For better visualization, we provide an enlargement of the small enclosed region in the large enclosed frame, connected by dashed lines. We set the array parameters for NA, SNA, ANA variants, and MISC as $N_1=8,\, N_2=92$, while for CP, CCP, TCA, CACis, CADis, SECA and RECA we choose $M=5,\,N=92$. The description of each array parameter can be found in the respective manuscripts cited throughout our paper.}
  \label{fig:arraycomparison}
 \end{figure*}

To that end, we consider our fractal scheme and first seek a small
generator array which meets the above specifications by solving the following problem
\begin{equation}
        \bb{S} = \underset{\bb{T}}{\arg\min}\; \abs{\bb{T}}\quad 
        \text{subject to}\quad\text{(R1)-(R5)}, A_\bb{T}\leq 20d,
        \tag{P1}
\end{equation}
where $A_\bb{T}$ denotes the physical aperture of the arrays. We seek an array with the fewest elements to  promote economy. Problem (P1) can be solved by e.g. performing a naive exhaustive search over all $2^{20}$ possible arrays. This of course is possible only for small scale. A solution of problem (P1) is
\begin{equation*}
    \bb{S} = [0\;\;1\;\;2\;\; 4\;\; 7\;\; 10\;\; 13\;\; 16\;\; 18\;\;19\;\;20],
\end{equation*}
which exhibits fragility of $F_\bb{S}=0.27$ and coupling leakage of $\mathcal{L}_\bb{S}=0.3$. In addition, we solve (P1) ignoring the symmetry constraint, leading to the following solution:
\begin{equation*}
    \bb{G} = [0\;\;1\;\; 3\;\; 5\;\; 11\;\; 13\;\; 17\;\;18\;\;19\;\;20].
\end{equation*}
We later use the array $\bb{G}$ for additional comparison.

Now, we apply our recursive scheme \eqref{eq:fractal} with $r=2$ where we utilize $\bb{S}$ and $\bb{G}$ as generators to create the sparse fractal arrays $\bb{S}_2$ and $\bb{G}_2$ respectively. As proven throughout the paper, the array construction \eqref{eq:fractal} guarantees that $\bb{S}_2$ meets requirements (R1)-(R6) as well as $\bb{G}_2$ excluding symmetry. This is seen in Fig.\,\ref{fig:arraycomparison} where we can observe that both $\bb{S}_2$ and $\bb{G}_2$ are located inside the feasible region. Thus, the fractal design allows us to create large sparse arrays which demonstrate low mutual coupling and low fragility simultaneously. 

To further study the fractal arrays we analyse their performance in DOA estimation, using the coarray
MUSIC algorithm, in comparison with commonly used sparse arrays: nested arrays (NA), super nested arrays (SNA), extended coprime array (CP) and complementary coprime arrays (CCP). We perform the comparison in small scale, i.e. small apertures with few elements, and in large scale while we examine three performance aspects - mutual coupling, robustness to sensor failures and sensitivity to noise. To evaluate the latter we use increasing levels of signal-to-noise ratio (SNR), while for mutual coupling we rely on the model described earlier with increasing values of $\abs{c_1}$. To test robustness to sensor failures, we assume each sensor fails
independently with probability $p$ and we assess performance as a function of $p$. For the small scale scenario, we consider $K=20$ sources with unit-amplitudes and normalized DOA $\bar{\theta}_i\triangleq \sin(\theta_i)/2\in[-0.5,\,0.5]$ equally spaced in the range $[-0.45,\,0.45]$. In addition, we set the array parameters of NA and SNA as $N_1=N_2=4$ \cite{liu2016super}, whereas the parameters of CP and CPP \cite{wang2019hole} are $M=3,\, N=4$. We use a similar setup for the large scale scenario but assume $K=400$ unit-amplitude sources with normalized  DOA equally-distributed in the aforementioned range. The parameters of NA and SNA are $N_1=8,\,N_2=92$ while for CP and CPP we use $M=5,\,N=92$. A summary of the properties of the arrays tested is given in Table\,\ref{tab:arrayparams}.    
The number of snapshots is 1000 for all cases. We assess each array by applying coarray MUSIC to compute the estimated source directions $\hat{\bar{\theta}}_i$ and calculating the root-mean squared error RMSE=$\big(\sum_{i=1}^K (\bar{\theta}_i-\hat{\bar{\theta}}_i)^2/K\big)^{1/2}$
averaged over 500 Monte-Carlo runs. Note that as $p$ increases more and more sensors malfunction, compromising the array structure, so that MUSIC might fail to yield any estimation. Therefore, as in \cite{liu2019robustness}, we only collect and average the instances in which coarray MUSIC was able to produce an estimation. To implement coarray MUSIC we utilize the online available code \cite{MUSICcode} used in, e.g., \cite{liu2016super, liu2019robustness}.

 \begin{table}
 \footnotesize
  \begin{center}
    \begin{tabular}{|c|c|c|c|}
      \hline % <-- Toprule here
      \rowcolor{LightBlue}
      \textbf{Array} & \textbf{\#Sensors} & \textbf{Fragility} & \textbf{Coupling Leakage} \\
      \hline % <-- Midrule here
      \rowcolor{LightGreen}
     \multicolumn{4}{|c|}{Small Scale} \\ 
     \hline
     NA       & 8  & 1    & 0.32 \\
     SNA      & 8  & 1    & 0.25 \\
     CP       & 9  & 2/3  & 0.26 \\
     CCP      & 11 & 0.45 & 0.30 \\  
     $\bb{S}$ & 11 & 0.27 & 0.30 \\ 
     $\bb{G}$ & 10 & 0.30  & 0.31 \\ 
     \hline
     \rowcolor{LightGreen}
     \multicolumn{4}{|c|}{Large Scale} \\ 
     \hline
     NA         & 100  & 1     & 0.14 \\
     SNA        & 100  & 1     & 0.01 \\
     CP         & 101  & 0.95  & 0.11 \\
     CCP        & 105  & 0.93  & 0.14 \\  
     $\bb{S}_2$ & 121  & 0.03  & 0.30 \\ 
     $\bb{G}_2$ & 100  & 0.09  & 0.31 \\
     $\bb{S}_3$ & 1331 & 0.006 & 0.30 \\
     $\bb{G}_3$ & 1000 & 0.027 & 0.31 \\
     \hline % <-- Bottomrule here
    \end{tabular}
     \caption{Array Properties.}
     \label{tab:arrayparams}
  \end{center}
%   \vspace{-0.5cm}
\end{table}

 The simulation results, shown in Fig.\,\ref{fig:experiments}, demonstrate that the optimized arrays $\bb{S}$ and $\bb{G}$ are more robust than the other arrays in all scenarios, achieving the lowest errors. As seen, when the element coupling is low, $\bb{S}$ and $\bb{G}$ lead to small errors which increase as $\abs{c_1}$ increases until high mutual coupling is reached and their performance is comparable to that obtained by the other arrays. Moreover, $\bb{S}$, $\bb{G}$ and CCP are less sensitive to noise than the alternative arrays as they obtain low errors in low SNR regions and achieve considerably better performance as SNR increases. Finally, $\bb{S}$, $\bb{G}$ outperform the other configurations, including CCP, when the probability for sensor failure exceeds a certain point, demonstrating the robustness of the optimized arrays.  
 
 While we showed that $\bb{S}$ and $\bb{G}$ are superior to other investigated arrays, our actual goal is to show this performance is maintained when we scale up the arrays to create $\bb{S}_2$ and $\bb{G}_2$. To that end, we examine the results on the bottom of Fig.\,\ref{fig:experiments}. We observe that the fractal arrays yield low errors when the mutual coupling is low, while in high coupling regions their performance is slightly degraded yet comparable to that of the other arrays. As clearly seen, both $\bb{S}_2$ and $\bb{G}_2$ surpass the alternative arrays for different probabilities of failure and in different SNR regimes, leading to considerably lower errors. Above a probability of failure of $p=0.2$, the fractal arrays lead to acceptable errors while for the other arrays MUSIC failed to produce DOA estimations. These results coincide with the properties given in Table\,\ref{tab:arrayparams} where we see that while all arrays exhibit relatively low coupling leakage, our fractal arrays are dramatically more robust than the other arrays.
  
 These experiments prove the effectiveness and simplicity of the proposed approach for constructing large sparse arrays while considering diverse specifications. We can continue and enlarge our generators further to create sparse fractal arrays with thousands of elements, as shown in Table\,\ref{tab:arrayparams}, which are expected to be required by applications such as massive MIMO and ultrasound imaging in the forthcoming years. 
 
%  As current applications such as massive MIMO use hundreds of sensors and aim at utilizing thousands of sensors in the near future, our fractal design facilitate the construction of sparse arrays with extra large apertures.   

 \begin{figure*}[ht]
 \centering
 \includegraphics[trim={1.5cm 3.5cm 1cm 3.5cm},clip,height = 8cm, width = 0.9\linewidth]{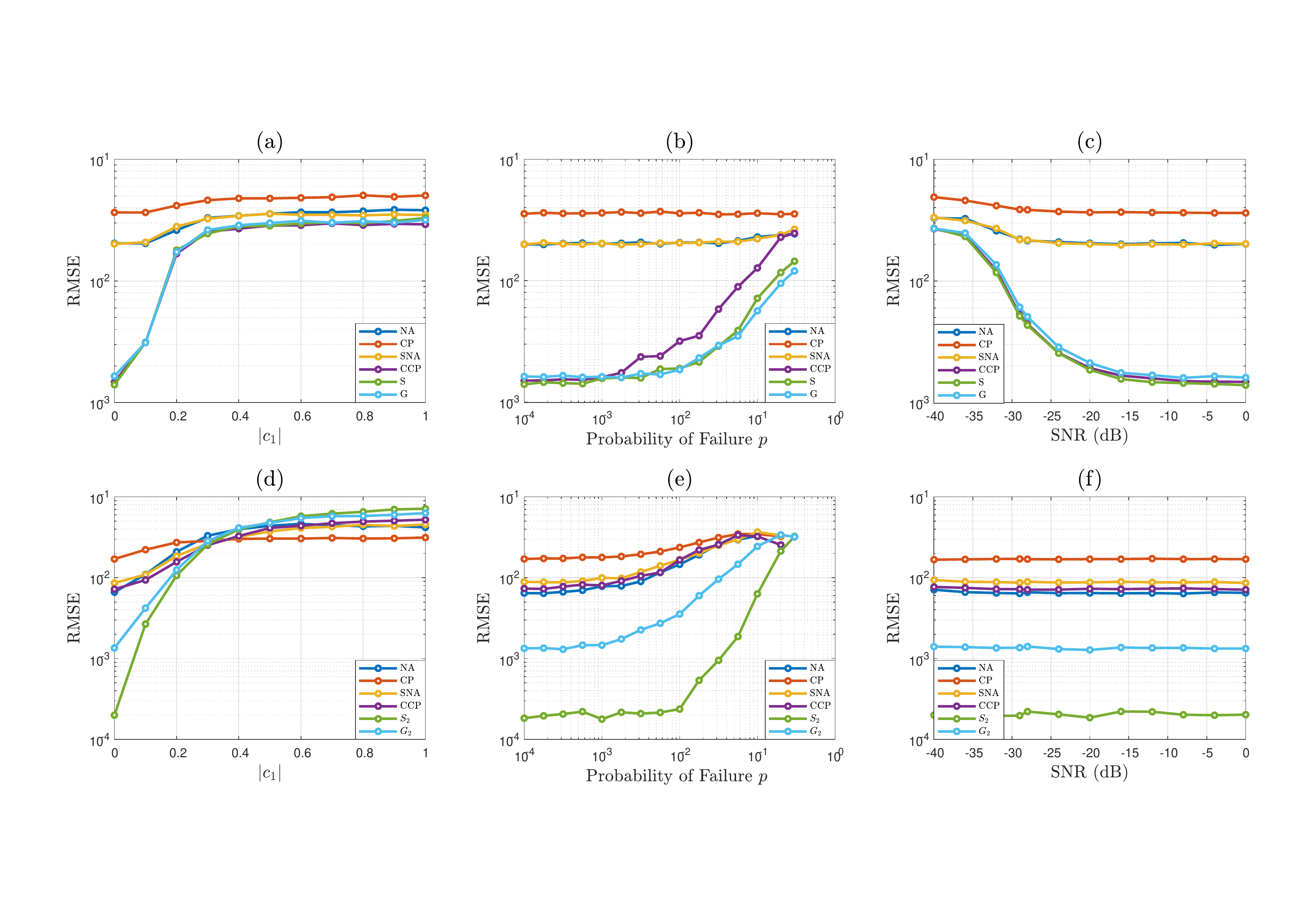}
\caption{\textbf{Experiments.} Performance comparison of the arrays described in Section\,\ref{sec:experiments} for the task of DOA estimation in three different scenarios: (left) increasing mutual coupling, (middle) increasing probability to sensor failure and (right) increasing SNR. The results on the top line were obtained with small scale arrays while those on the bottom were attained using large scale arrays whose properties are given in Table\,\ref{tab:arrayparams}. For (a), (b), (d) and (e) we used SNR=0dB.}
  \label{fig:experiments}
 \end{figure*}

% \begin{figure}
%  \centering
%  \includegraphics[trim={3.5cm 6.5cm 3cm 6cm},clip,height = 4cm, width = 1\linewidth]{Figures/Optimal.pdf}
% \caption{\textbf{Optimal Generators.} (a), (b) and (c) are the three optimal solutions with respect to the array specifications described in Section\,\ref{sec:experiments} and to the number of sensors.    }
%   \label{fig:solution}
%  \end{figure}

% \vspace{-5 pt}
\section{Conclusion}
\label{sec:conclude}
The design of large sparse arrays poses a major challenge. Various sparse geometries have been proposed over the last decades. However, most of these designs focus on certain aspects of the array while ignoring or being indifferent to other important properties. Incorporating all desired design criteria leads to combinatorial problems which currently cannot be solved efficiently in large scale.    

In this paper, we introduce a fractal scheme in which we use a sparse array as a generator and we expand it recursively according to its difference coarray. We proved that for an appropriate choice of the generator, the proposed design creates sparse fractal arrays with increased degrees of freedom, i.e., large difference coarrays. Thus, we can extend any known sparse configuration to an arbitrarily large array. Moreover, we presented a detailed analysis of the fractal arrays with respect to several important array characteristics. The analysis showed that fractal arrays inherit from their generators properties such as symmetry, array economy, mutual coupling and robustness to sensor failures. The array weight function and beampattern can also be easily derived from the generator. In addition, we presented a generalized fractal scheme that allows to combine different sparse geometries in which the number of sensors can grow moderately with the array order. 

Finally, we perform numerical experiments to demonstrate the practicality of the proposed fractal scheme. We outline a representative design plan which requires the array to be symmetric and robust to sensor failures while exhibiting low mutual coupling. As shown, most popular sparse configurations do not meet these requirements as they were designed to achieve high DOF which increases their fragility at the same time. We then constructed fractal arrays using our design scheme which display low coupling leakage and low fragility simultaneously. We evaluate the performance of our fractal arrays in comparison with several common sparse arrays, showing their superiority in various scenarios. Thus, this work provides a simple and scalable fractal approach for designing large scale sparse arrays with multiple properties.              

\vspace{-5 pt}
\appendices
\section{Proof of Theorem\,\ref{theo:multifree}}
\label{app:mutlifree}
We prove the theorem by induction.
\begin{itemize}
\item \textbf{Base} ($k=1$): In this case  $\bb{M}_1=\bb{G}_1$. Hence, $\bb{D}_1=\bb{D}_{\bb{G}_1}$ and it can be written as 
\begin{equation*}
\bb{D}_1=\bb{D}=\left[-\frac{M-1}{2},\frac{M-1}{2}\right],
\end{equation*}
where $M=\abs{\bb{D}_{\bb{G}_1}}$ since we assume that $\bb{D}_{\bb{G}_1}$ is hole-free.
\item \textbf{Assumption} ($k=r$): $\bb{D}_r$ is a hole-free array given by 
\begin{equation*}
    \bb{D}_r=\left[-\frac{M_r-1}{2},\frac{M_r-1}{2}\right],
\end{equation*}
where $M_r=\prod\limits_{i=1}^r\abs{\bb{D}_{\bb{G}_i}}$.
\item \textbf{Step} ($k=r+1$): The difference coarray $\bb{D}_{\bb{G}_{r+1}}$ of $\bb{G}_{r+1}$ is assumed to be hole-free, hence,
\begin{equation*}
\bb{D}_{\bb{G}_{r+1}}=\left[-\frac{\abs{\bb{D}_{\bb{G}_{r+1}}}-1}{2},\frac{\abs{\bb{D}_{\bb{G}_{r+1}}}-1}{2}\right].
\end{equation*}
By definition of the difference coarray, we have
{\small
\begin{align*}
\bb{D}_{r+1}&\defeq \{k-l:\, k,l\in\bb{M}_{r+1}\} \\
&=\{s+uM_r-(t+vM_r): s,t\in\bb{M}_r,u,v\in\bb{G}_{r+1}\} \\
&=\{(s-t)+(u-v)M_r: s,t\in\bb{M}_r,u,v\in\bb{G}_{r+1}\} \\
&=\{m+nM_r:\, m\in\bb{D}_r,\,n\in\bb{D}_{\bb{G}_{r+1}}\}.
\end{align*}}
\hspace{-9pt}
Since $\bb{D}_{\bb{G}_{r+1}}$ is hole-free and $M_r=\abs{\bb{D}_r}$, we have that $\bb{D}_{r+1}$ consists of $l=\abs{\bb{D}_{\bb{G}_{r+1}}}$ consecutive replicas of $\bb{D}_r$:
\begin{equation*}
    \bb{D}_{r+1} = \underbrace{[\bb{D}_r\;\bb{D}_r\;\dots\;\bb{D}_r]}_{l\text{ times}}.
\end{equation*}

By our assumption $\bb{D}_r$ is hole-free, implying that $\bb{D}_{r+1}$ is hole-free and is given by
\begin{equation*}
    \bb{D}_{r+1}=\left[-\frac{M_{r+1}-1}{2},\frac{M_{r+1}-1}{2}\right],
\end{equation*}
where 
\begin{align*}
   M_{r+1}\defeq\abs{\bb{D}_{r+1}}&=\abs{\bb{D}_r}\cdot\abs{\bb{D}_{\bb{G}_{r+1}}}=M_r\cdot\abs{\bb{D}_{\bb{G}_{r+1}}} \\ &=\Big(\prod\limits_{i=1}^r\abs{\bb{D}_{\bb{G}_i}}\Big)\cdot\abs{\bb{D}_{\bb{G}_{r+1}}}=\prod\limits_{i=1}^{r+1}\abs{\bb{D}_{\bb{G}_i}},
\end{align*}
completing the proof.
\end{itemize}

\vspace{-5 pt}
\section{Proof of Theorem\,\ref{theo:multilarge}}
\label{app:multilarge}
Denoting the central ULA of $\bb{D}_r$ by $\bb{U}_r$, we first prove by induction that
\begin{equation*}
    \left[-\frac{M_r-1}{2},\frac{M_r-1}{2}\right]\subseteq \bb{U}_r,
\end{equation*}
where $M_r\triangleq\prod\limits_{i=1}^r\abs{\bb{U}_{\bb{G}_i}}$. In particular, $\abs{\bb{U}_r}=\mathcal{O}(M_r)$.
\begin{itemize}
\item \textbf{Base} ($k=1$):
In this case  $\bb{M}_1=\bb{G}_1$. Hence, $\bb{U}_1=\bb{U}_{\bb{G}_1}$ and it can be written as 
\begin{equation*}
\bb{U}_1=\bb{U}=\left[-\frac{M-1}{2},\frac{M-1}{2}\right],
\end{equation*}
where $M=\abs{\bb{U}_{\bb{G}_1}}$ since by definition $\bb{U}_{\bb{G}_1}$ is hole-free and symmetric.
\item \textbf{Assumption} ($k=r$): Assume that
\begin{equation*}
    \left[-\frac{M_r-1}{2},\frac{M_r-1}{2}\right]\subseteq \bb{U}_r.
\end{equation*}
\item \textbf{Step} ($k=r+1$): We define the following sets 
\begin{align*}
    &\bb{T}_r\triangleq  \left[-\frac{M_r-1}{2},\frac{M_r-1}{2}\right], \\
    &\bb{Y}_r\triangleq \{m+nM_r:\, m\in\bb{T}_r,\,n\in\bb{U}_{\bb{G}_{r+1}}\}, \\
    &\bb{V}_r\triangleq \{m+nM_r:\, m\in\bb{U}_r,\,n\in\bb{U}_{\bb{G}_{r+1}}\}.
\end{align*}
Notice that $\bb{Y}_r$ can be written in explicit form as 
\begin{equation*}
        \bb{Y}_r=\left[-\frac{M_{r+1}-1}{2},\frac{M_{r+1}-1}{2}\right],
\end{equation*}
where % Change 3
\begin{align*}
    M_{r+1}=\abs{\bb{T}_r}\cdot\abs{\bb{U}_{\bb{G}_{r+1}}}
    =M_r\cdot\abs{\bb{U}_{\bb{G}_{r+1}}}
    =\prod\limits_{i=1}^{r+1}\abs{\bb{U}_{\bb{G}_i}}.
\end{align*}
In addition, both $\bb{U}_r$ and $\bb{U}_{\bb{G}_{r+1}}$ are symmetric and hole-free arrays where $M_r\leq\abs{\bb{U}_r}$. Therefore, by the construction of $\bb{V}_r$, we have that
$\bb{U}_r\subseteq \bb{V}_r$ and $\bb{V}_r$ is symmetric and hole-free. Similar to proof of Theorem\,\ref{theo:multifree}, we can express the difference coarray
\begin{align*}
\bb{D}_{r+1}=\{m+nM_r:\, m\in\bb{D}_r,\,n\in\bb{D}_{\bb{G}_{r+1}}\}.
\end{align*}
As $\bb{U}_r\subseteq\bb{D}_r$ and $\bb{U}_{\bb{G}_{r+1}}\subseteq \bb{D}_{\bb{G}_{r+1}}$, we get that $\bb{V}_r\subseteq \bb{D}_{r+1}$. This suggests that $\bb{V}_r\subseteq \bb{U}_{r+1}$ since by the definition of the central ULA, $\bb{U}_{r+1}$ is the longest symmetric hole-free array in the difference coarray.

By the induction assumption, $\bb{T}_r\subseteq\bb{U}_r$ implying that $\bb{Y}_r\subseteq \bb{V}_r$, which in turn leads to
\begin{equation*}
    \bb{Y}_r=\left[-\frac{M_{r+1}-1}{2},\frac{M_{r+1}-1}{2}\right]\subseteq \bb{U}_{r+1},
\end{equation*}
since $\bb{Y}_r\subseteq\bb{V}_r\subseteq \bb{U}_{r+1}$. Thus, we obtain that 
\begin{equation*}
    \abs{\bb{U}_{r+1}}\geq \Bigg|\left[-\frac{M_{r+1}-1}{2},\frac{M_{r+1}-1}{2}\right]\Bigg|=M_{r+1}.
\end{equation*}
\end{itemize}
Now, since $ \bb{U}_r \subseteq \bb{D}_r$ we have that
\begin{equation*}
\abs{\bb{D}_r}\geq \abs{\bb{U}_r}\geq M_r=\prod\limits_{i=1}^r\abs{\bb{U}_{\bb{G}_i}}.
\end{equation*}
Finally, recall that $\abs{\bb{U}_{\bb{G}_i}}=\mathcal{O}(\abs{\bb{G}_i}^2)$ for all $1\leq i\leq r$ and $N\leq \prod\limits_{i=1}^r\abs{\bb{G}_i}$, hence, $\abs{\bb{D}_r}=\mathcal{O}(M_r)=\mathcal{O}(N^2)$
% \begin{equation*}
% \abs{\bb{D}_r}=\mathcal{O}(M_r)=\mathcal{O}(N^2),
% \end{equation*}
which completes the proof. 

% \newpage
% \section*{Acknowledgment}
% \FloatBarrier
\bibliographystyle{IEEEtran}
\bibliography{IEEEabrv,refs}

\end{document}